\documentclass[11 pt]{article}                  
\usepackage{graphicx}
\usepackage[all]{xy}
\usepackage{color}
\usepackage{latexsym}
\usepackage{amsmath,amstext}
\usepackage{amssymb}
\usepackage{epsfig}
\usepackage{graphics}
\usepackage{euscript}
\usepackage{ifthen}
\usepackage{enumerate}
\usepackage{wrapfig}
\usepackage{amsthm}
\usepackage{multicol}
\usepackage{paralist}
\usepackage{verbatim}
\usepackage[colorlinks=true, urlcolor=blue, linkcolor=blue, citecolor=blue]{hyperref}
\usepackage[margin=1in]{geometry}
\usepackage{placeins}
\usepackage[linesnumbered,lined,commentsnumbered,ruled]{algorithm2e}
\usepackage{caption}
\usepackage{subcaption}
\usepackage{longtable}
\usepackage{lscape}
\usepackage{fancyhdr}
\usepackage{listings}
\usepackage{mathtools}
\usepackage{appendix} 
\usepackage{float}
\usepackage{textcomp}
\usepackage[inline]{enumitem}
\usepackage{mathabx}

\usepackage{fontenc}
\usepackage{imakeidx}
\makeindex

\usepackage[utf8]{inputenc}
\usepackage{csquotes}
\usepackage{xparse}
\usepackage[dvipsnames]{xcolor}
\definecolor{VIOLET}{rgb}{0.2, 0.5, 0.478}

\usepackage[normalem]{ulem}
\theoremstyle{plain}
\newtheorem{theorem}{Theorem}[section]

\newtheorem{lemma}[theorem]{Lemma}
\newtheorem{proposition}[theorem]{Proposition}

\theoremstyle{definition}
\newtheorem{definition}[theorem]{Definition}
\newtheorem{remark}[theorem]{Remark}

\numberwithin{equation}{section}
\numberwithin{equation}{section}

\newcommand{\R}{\mathbb{R}}

\newcommand{\C}{\mathbb{C}}

\begin{document}

\title{Embeddability of centrosymmetric matrices capturing
the double-helix structure in natural and synthetic DNA}
%

\author{Muhammad Ardiyansyah        \and
        Dimitra Kosta \and 
        Jordi Roca-Lacostena}




\maketitle
\begin{abstract}
 \noindent In this paper, we discuss the embedding problem for centrosymmetric matrices, which are higher order generalizations of the matrices occurring in strand symmetric models. These models capture the substitution symmetries arising from the double helix structure of the DNA. Deciding whether a transition matrix is embeddable or not enables us to know if the observed substitution probabilities are consistent with a homogeneous continuous time substitution model, such as the Kimura models, the Jukes-Cantor model or the general time-reversible model. On the other hand, the generalization to higher order matrices is motivated by the setting of synthetic biology, which works with different sizes of genetic alphabets.
\end{abstract}

\smallskip

\smallskip

\section{Introduction}

Phylogenetics is the study of evolutionary relationships among species that aims to infer the evolutionary history among them. In order to model evolution, we consider a phylogenetic tree, that is a directed acyclic graph depicting the evolutionary relationships amongst a selected set of taxa. Phylogenetic trees consist of vertices and edges. Vertices represent biological entities, while edges between vertices represent the evolutionary processes between the taxa.

In order to describe the real evolutionary process along an edge of a phylogenetic tree, one often assumes that the evolutionary data occurred following a Markov process. A Markov process is a random process in which the future is independent of the past, given the present. Under this Markov process, transitions between $n$ states given by conditional probabilities are presented in a $n\times n$ Markov matrix $M$, namely a square matrix whose entries are nonnegative and rows sum to one. A well-known problem in probability theory is the so-called embedding problem which was initially posed by Elfving \cite{Elfving}. The embedding problem asks whether given a Markov matrix $M$, one can find a real square matrix $Q$ with rows summing to zero and non-negative off-diagonal entries, such that $M=\exp(Q)$. The matrix $Q$ is  called a Markov generator.

 In the complex setting, the embedding problem is completely solved by \cite{Higham}; a complex matrix $A$ is embeddable if and only if $A$ is invertible. However, as our motivation arises from molecular models of evolution we are interested in the embedding problem over the real numbers, so from now on we will denote by $M$ a real Markov matrix. It was shown by Kingman \cite{kingman1962imbedding} that if an $n \times n$ real Markov matrix $M$ is embeddable, then the matrix $M$ has $\det{M}>0$.  Moreover, in the same work by Kingman it was shown that $\det{M}>0$ is a necessary and sufficient condition for a $2\times 2$ Markov matrix $M$ to be embeddable. For $3\times 3$ Markov matrices a complete solution of the embedding problem is provided in a series of papers \cite{Cuthbert, Johansen, Carette1995CharacterizationsOE, ChenChen}, where the characterisation of embeddable matrices depends on the Jordan decomposition of the Markov matrix. For $4 \times 4$ Markov matrices the embedding problem is completely settled in a series of papers \cite{casanellas2020embeddability, casanellas2020embedding,RocaFernandez}, where similarly to the $3\times 3$ case the full characterisation of embeddable matrices is distinguished into cases depending on the Jordan form of the Markov matrices. 

For the general case of $n \times n$ Markov matrices, there are several results; some presenting necessary conditions \cite{Elfving, kingman1962imbedding, Runnenberg}, while others sufficient conditions \cite{Cuthbert, Fuglede, Goodman, davies2010embeddable} for embeddability of Markov matrices. Moreover, the embedding problem has been solved for special $n\times n$ matrices with a biological interest such as equal-input and circulant matrices \cite{BaakeSumner}, group-based models \cite{ardiyansyah2021model} and time-reversible models \cite{Jia}. Despite the fact that there is no theoretical explicit solution for the embeddability of general $n\times n$ Markov matrices, there are results \cite{casanellas2020embedding} that enable us to decide whether a $n\times n$ Markov matrix with distinct eigenvalues is embeddable or not. This is achieved by providing an algorithm that outputs all Markov generators of such a Markov matrix \cite{casanellas2020embedding, Roca}.

In this paper, we focus on the embedding problem for $n \times n$ matrices that are symmetric about their center and are called centrosymmetric matrices (see Definition~\ref{Definition: cs matrix}). We also study a variation of the famous embedding problem called model embeddability\index{model embeddability}, where apart from the requirement that the Markov matrix is the matrix exponential of a rate matrix, we additionally ask that the rate matrix follows the model structure. For instance, for centrosymmetric matrices, model embeddability means that the rate matrix is also centrosymmetric.

The motivation for studying centrosymmetric matrices comes from evolutionary biology, as the most general nucleotide substitution model when considering both DNA strands admits any $n\times n$ centrosymmetric Markov matrix as a transition matrix, where $n$ is the even number of nucleotides. For instance, by considering the four natural nucleotides \texttt{A-T, C-G} we arrive at the strand symmetric model, a well-known phylogenetic model whose substitution probabilities reflect the symmetry arising from the complementarity between the two strands that the DNA is composed of (see \cite{CasanellasSullivant}). 
In particular, a strand
symmetric model for DNA must have the following equalities of probabilities
in the root distribution:
\begin{eqnarray}
\label{symmetries}
\pi_{\texttt{A}} = \pi_{\texttt{T}} \text{ and }
\pi_{\texttt{C}} = \pi_{\texttt{G}}
\end{eqnarray}
and the following equalities of probabilities in the transition matrices $(\theta_{ij})$
\begin{eqnarray*}
\theta_{\texttt{A} \texttt{A}} = \theta_{\texttt{T} \texttt{T}}, \theta_{\texttt{A} \texttt{C}} = \theta_{\texttt{T} \texttt{G}}, \theta_{\texttt{A} \texttt{G}} = \theta_{\texttt{T} \texttt{C}}, \theta_{\texttt{A} \texttt{T}} = \theta_{\texttt{T} \texttt{A}},\\
\theta_{\texttt{C} \texttt{A}} = \theta_{\texttt{G} \texttt{T}} , \theta_{\texttt{C} \texttt{C}} = \theta_{\texttt{G} \texttt{G}}, \theta_{\texttt{C} \texttt{G}} = \theta_{\texttt{G} \texttt{C}}, \theta_{\texttt{C} \texttt{T}} = \theta_{\texttt{G} \texttt{A}}.
\end{eqnarray*}

Therefore, the corresponding transition matrices of this model are $4\times 4$ centrosymmetric matrices, usually called strand symmetric Markov matrices in this context. In the strand symmetric model there are less restrictions on the way genes mutate from ancestor to
child compared to other widely known molecular models of evolution. In fact, special cases of the strand symmetric model are the group-based phylogenetic models such as the Jukes-Cantor (JC)\index{Jukes-Cantor (JC) model }
model, the Kimura 2-parameter (K2P)\index{Kimura 2-parameter model (K2P)} and Kimura 3-parameter (K3P)\index{Kimura 3-parameter model (K3P)} models. The algebraic structure of strand symmetric models was initially studied in \cite{CasanellasSullivant}, where it was argued that strand symmetric models capture more biologically meaningful features of real DNA sequences than the commonly used group-based models, as for instance, in any group-based model, the stationary distribution of bases for a single species is always the uniform distribution, while computational evidence in \cite{YapPachter} suggests that the
stationary distribution of bases for a single species is rarely uniform, but must
always satisfy the symmetries (\ref{symmetries}) arising from nucleotide complementarity, as assumed by the strand symmetric model. 

In this article, we also explore higher order centrosymmetric matrices  for which $n > 4$, which is justified by the use of synthetic nucleotides.  
One of main goals of synthetic biology is to expand the genetic alphabet to include an unnatural or synthetic base pair. The more letters in a genetic system could possibly lead to an increased potential for retrievable information storage and bar-coding and combinatorial tagging \cite{benner2005synthetic}. Naturally the four-letter genetic alphabet consists of just two pairs, \texttt{A-T} and \texttt{G-C}. 
In 2012, a genetic system comprising of three base pairs was introduced in \cite{malyshev2012efficient}. In addition to the natural base pairs, the third, unnatural or synthetic base pair \texttt{5SICS-MMO2} was proven to be functionally equivalent to a natural base pair. Moreover, when it is combined with the natural base pairs,  \texttt{5SICS-MMO2} provides a fully functional six-letter genetic alphabet. Namely, six-letter genetic alphabets can be copied \cite{yang2007enzymatic}, polymerase chain reaction (PCR)-amplified and sequenced \cite{sismour2004pcr,yang2011amplification}, transcribed to six-letter RNA and back to six-letter DNA \cite{leal2015transcription}, and used to encode proteins with added amino acids \cite{bain1992ribosome}. This biological importance and relevance of the above six-letter genetic alphabets motivates us to particularly study the $6\times 6$ Markov matrices describing the probabilities of changing base pairs in the six-letter genetic system in Section~\ref{sect:6x6}. When considering both DNA strands, each substitution is observed twice due to the complementarity between both strands, and hence the resulting transition matrix is centrosymmetric.

Moreover there are other synthetic analogs to natural DNA which justify studying centrosymmetric matrices for $n>6$. For instance, hachimoji DNA is a synthetic DNA that uses four synthetic nucleotides \texttt{B, Z, P, S}  in addition to the four natural ones \texttt{A,C, G, T}. With the additional four synthetic ones, hachimoji DNA forms four types of base pairs, two of which are unnatural: \texttt{P} binds with \texttt{Z} and \texttt{B} binds with \texttt{S}. The complementarity between both strands of the DNA implies that the transition matrix is centrosymmetric. Moreover, the research group responsible for the hachimoji DNA system had also studied a synthetic DNA analog system that used twelve different nucleotides, including the four found in DNA (see \cite{Benner12nuclei}).
Although the biological models which motivate the study of centrosymmetric matrices in this paper require $n$ to be an even number due to the double-helix structure of DNA, in Section~\ref{sect:centrosymmetric}, we include the case of $n$ being odd for completeness.

 Apart from embeddability, namely existence of Markov generators, it is also natural to ask about uniqueness of a Markov generator which is called the rate identifiability problem. Identifiability is a property which a model must satisfy in order for precise statistical inference to be possible.
A class of phylogenetic models is identifiable if any two models in the class produce different data distributions. In this article, we further develop the results on rate identifiability of the Kimura two parameter model 
\cite{casanellas2020embeddability} to study rate identifiability for strand symmetric models. We also show that there are embeddable strand symmetric Markov matrices with non identifiable rates, namely the Markov generator is not unique. {\color{black} Moreover, we show that {\color{black}strand symmetric} Markov matrices are not generically identifiable, that is, there exists a positive measure subset of {\color{black}strand symmetric} Markov matrices containing embeddable matrices whose rates are not identifiable.}

This paper is organised as following. In Section~\ref{sect:prelim}, we introduce the basic definitions and results on embeddability. {\color{black} In Section~\ref{sect:embed}, we give a characterisation for a $4 \times 4$ centrosymmetric Markov matrix $M$ (also known as a strand symmetric Markov matrix) with four distinct real nonnegative eigenvalues to be embeddable providing necessary and sufficient conditions in Theorem~\ref{theorem:criteria SSM-embeddability}, while we also discuss their rate
identifiability property in Proposition~\ref{prop:rate_idendifiable}. }    Moreover in Section~\ref{sect:volume}, {\color{black}using the conditions of our main result Theorem~\ref{theorem:criteria SSM-embeddability}}, we compute the relative volume of all {\color{black} strand symmetric} Markov matrices relative to the {\color{black} strand symmetric} Markov matrices with positive eigenvalues and $\Delta >0$, as well as the relative volume of all {\color{black} strand symmetric} Markov matrices relative to the {\color{black} strand symmetric} Markov matrices with four distinct eigenvalues and $\Delta >0$. {\color{black}We also compare the results on relative volumes obtained using our method with the algorithm suggested in \cite{casanellas2020embedding} to showcase the advantages of our method.} 
In Section~\ref{sect:centrosymmetric}, we study higher order centrosymmetric matrices and motivate their use in Section~\ref{sect:6x6} by exploring the case of synthetic nucleotides where the phylogenetic models admit $6 \times 6$ centrosymmetric mutation matrices. Finally, Section~\ref{discussion} discusses implications and possibilities for future work.

\section{Preliminaries}
\label{sect:prelim}

In this section we will introduce the definitions and results that will be required throughout the paper. We will denote by $M_n(\mathbb{K})$ the set of $n \times n$ square matrices with entries in the field $\mathbb{K}= \mathbb{R} \text{ or } \mathbb{C}$. The subset of non-singular matrices in $M_n(\mathbb{K})$ will be denoted by $GL_{n}(\mathbb{K})$. 

\begin{definition}
{\color{black}
We will call \textit{Markov\index{Markov matrix} (or transition\index{transition matrix}) matrices} the non-negative real square matrices with rows summing to one}.  \textit{Rate matrices\index{rate matrix}} are real square matrices with rows summing to zero and non-negative off-diagonal entries.
    
\end{definition}

{\color{black} In this paper, we are focusing on a subset of Markov matrices called centrosymmetric Markov matrices\index{centrosymmetric (CS) Markov matrix}.}
\begin{definition}\label{Definition: cs matrix}
{\color{black}A real  $n\times n$ matrix $A=(a_{i,j})$} is said to be \textit{centrosymmetric (CS)} if $$a_{i,j}=a_{n+1-i,n+1-j}$$
for every $1\leq i,j\leq n$.
\end{definition}

 {\color{black}Definition \ref{Definition: cs matrix} reveals} that a CS matrix is nothing more than a square matrix which is symmetric about its center. This class of matrices has been previously studied, for instance, in \cite[page 124]{aitken2017determinants} and \cite{weaver1985centrosymmetric}. Examples of CS matrices for $n=5$ and $n=6$, {\color{black} are the following two matrices respectively:}
    \begin{align*}
    \begin{pmatrix}
    a_{11}&a_{12}&a_{13}&a_{14}& a_{15}\\
    a_{21}&a_{22}&a_{23}&a_{24}&a_{25}\\
    a_{31}&a_{32}&a_{33}&a_{32}&a_{31}\\
    a_{25}&a_{24}&a_{23}&a_{22}&a_{21}\\
    a_{15}&a_{14}&a_{13}&a_{12}&a_{11}\\
    \end{pmatrix} \qquad  \text{ and } \qquad  \begin{pmatrix}
    a_{11}&a_{12}&a_{13}&a_{14}& a_{15}&a_{16}\\
    a_{21}&a_{22}&a_{23}&a_{24}&a_{25}&a_{26}\\
    a_{31}&a_{32}&a_{33}&a_{34}&a_{35}&a_{36}\\
    a_{36}&a_{35}&a_{34}&a_{33}&a_{32}&a_{31}\\        a_{26}&a_{25}&a_{24}&a_{23}&a_{22}&a_{21}\\
    a_{16}&a_{15}&a_{14}&a_{13}&a_{12}&a_{11}\\
    \end{pmatrix}.
\end{align*}
{\color{black}The class of CS matrices} plays an important role in the study of Markov processes since they are indeed transition matrices for some processes {\color{black}in evolutionary biology.} {\color{black}For instance, in \cite{kimura1957some}, centrosymmetric matrices are used to study the random assortment phenomena of subunits in chromosome division. Furthermore, in \cite{ schensted1958appendix}, the same centrosymmetric matrices appear as the transition matrices in the model of subnuclear segregation in the macronucleus of ciliates. Finally, the work \cite{iosifescu2014finite} examines a special case of random genetic drift phenomenon, which consists of a population
consisting of individuals that are able to produce a single type of gamete and the transition matrices of the associated Markov chain are given by centrosymmetric matrices. \\
}

{\color{black}The embedding problem is directly related to  the notions of matrix exponential and logarithm which we introduce for completeness below.}
\begin{definition}
We define the exponential\index{exponential $\exp(A)$ of matrix $A$} $\exp(A)$ of a matrix $A$, using the Taylor power series of the function $f(x)=e^x$, as
$$
\exp(A) = \sum_{k=0}^{\infty} \frac{A^k}{k!},$$
where $A^0= I_n$ and $I_n$ denotes the $n\times n$ identity matrix. If $A=P \;diag(\lambda_1,\dots,\lambda_n)\;P^{-1}$ is an eigendecomposition of $A$, then $\exp(A) = P \;diag(e^{\lambda_1},\dots,e^{\lambda_n})\;P^{-1}$.
Given a matrix $A\in M_n(\mathbb{K})$, a matrix $B\in M_n(\mathbb{K})$ is said to be a \textit{logarithm of}\index{logarithm of a matrix} $A$ if $\exp(B)=A$. If $v$ is an eigenvector corresponding to the eigenvalue $\lambda$ of $A$, then $v$ is an eigenvector corresponding to the eigenvalue $e^{\lambda}$ of $\exp(A)$.
\end{definition}

A Markov matrix $M$ is called \textit{embeddable\index{embeddable matrix}} if it can be written as the exponential of a rate matrix $Q$, namely $M= \exp(Q)$. Then any rate matrix $Q$ satisfying the equation $M= \exp(Q)$ is called a \textit{Markov generator\index{Markov generator} of }$M$.  \\

{\color{black} 
\begin{remark}
We should note here that embeddable Markov matrices occur when we assume a continuous time Markov chain, in which case the Markov matrices have the form 
$$
M= \exp(t Q),
$$
where $t \geq 0$ represents time and $Q$ is a rate matrix. However, in the rest of the paper, we assume that $t$ is
incorporated in the rate matrix $Q$.
\end{remark}
}

The existence of multiple logarithms is a direct consequence of the distinct branches of the logarithmic function in the complex field. 

\begin{definition}
Given $z\in \mathbb{C}\setminus \mathbb{R}_{\leq0}$ and $k\in \mathbb{Z}$, the $k$-th branch of the logarithm of $z$ is $ \log_k(z):= \log|z| + (Arg(z)+2\pi k)i$, where {\color{black} $\log$ is the logarithmic function on the real field} and $Arg(z)\in (-\pi,\pi)$ denotes the principal argument of $z$. The logarithmic function arising from the branch $\log_0(z)$ is called the principal logarithm of $z$ and is denoted as $\log(z)$.
\end{definition}

 It is known that if $A$ is a matrix with no negative eigenvalues, then there is a unique logarithm of $A$ all of whose eigenvalues are given by the principal logarithm of the eigenvalues of $A$ \cite[Theorem 1.31]{Higham}. We refer to this unique logarithm as the \textit{principal logarithm\index{principal logarithm $Log(A)$ of a matrix $A$} of} $A$, denoted by $Log(A)$.

{\color{black}By definition, the Markov generators of a Markov matrix $M$ are those logarithms of $M$ that are rate matrices. In particular they are real logarithms of $M$.} The following result enumerates all the real logarithms with rows summing to zero of any given Markov matrix with positive determinant and {\color{black}distinct} eigenvalues. Therefore, all Markov generators of such a matrix are necessarily of this form.

\begin{proposition}[{\cite[Proposition 4.3]{casanellas2020embedding}}] \label{Prop:Logenum}
Let $M=P\; diag \big( 1,\lambda_1,\dots,\lambda_t, \mu_1,\overline{\mu_1},\dots,\mu_s,\overline{\mu_s} \big) \; P^{-1}$ be an $n\times n$ Markov matrix with $P\in GL_n(\mathbb{C})$ and {\color{black} distinct} eigenvalues $\lambda_i\in \mathbb{R}_{>0}$ for $i=1,\dots,t$ and $\mu_j \in \{z\in \C: Im(z) > 0\}$ for $j=1,\dots,s$, all of them pairwise distinct. Then, a matrix $Q$ is a real logarithm of $M$ with rows summing to zero if and only if  $$Q=P\; diag \Big( 0, \log(\lambda_1),\dots,\log(\lambda_t), \log_{k_1}(\mu_1),\overline{\log_{k_1}(\mu_1)},\dots,\log_{k_s}(\mu_s),\overline{\log_{k_s}(\mu_s)} \Big) \; P^{-1}$$ for some $k_1,\dots,k_j\in \mathbb{Z}$.
\end{proposition}

\begin{remark}
\rm In particular, the principal logarithm of $M$ can be computed as $$Log(M) = P\; diag \Big( 0, \log(\lambda_1),\dots,\log(\lambda_t), \log(\mu_1),\log(\overline{\mu_1}),\dots,\log(\mu_s),\log (\overline{\mu_s}) \Big) \; P^{-1}.$$
\end{remark}

In this paper, we focus on the embedding problem for the class of centrosymmetric matrices. In Section~\ref{sect:embed}, we will first study the embeddability of $4\times 4$ centrosymmetric Markov matrices, which include the K3P, K2P and JC Markov matrices. In Section~\ref{sect:centrosymmetric} and Section~\ref{sect:6x6}, we will further study the embeddability of higher order centrosymmetric Markov matrices.

\section{Embeddability of {\color{black}$4\times 4$ centrosymmetric matrices}}
\label{sect:embed}
In this section, we begin our study by analyzing the embeddability of $4\times 4$ centrosymmetric matrices {\color{black} also known as strand symmetric matrices}. We will provide necessary and sufficient conditions for $4\times 4$ CS matrices to be embeddable. Moreover, we will discuss their rate identifiability problem as well.

Phylogenetic evolutionary models whose mutation matrices are $4\times 4$ centrosymmetric matrices are also called {\color{black}s}trand {\color{black}s}ymmetric Markov models. The transition matrices {\color{black}in the strand symmetric model} are assumed to have the form
\begin{align*}M=
    \begin{pmatrix}
    m_{11}&m_{12}&m_{13}&m_{14}\\
    m_{21}&m_{22}&m_{23}&m_{24}\\
    m_{24}&m_{23}&m_{22}&m_{21}\\
    m_{14}&m_{13}&m_{12}&m_{11}\\
    \end{pmatrix},
\end{align*}
where $$m_{11}+m_{12}+m_{13}+m_{14}=1=m_{21}+m_{22}+m_{23}+m_{24}\mbox{ and }m_{ij}\geq 0.$$ In biology, $4\times 4$ centrosymmetric Markov (rate) matrices are often referred {\color{black} to  as \textit{strand symmetric} Markov (rate) matrices. In this article, we will use $4\times 4$ centrosymmetric and  strand symmetric interchangeably. Recall that the K3P matrices are assumed to have the form
\begin{align*}
M=
    \begin{pmatrix}
   m_{11}&m_{12}&m_{13}&m_{14}\\
   m_{12}&m_{11}&m_{14}&m_{13}\\
   m_{13}&m_{14}&m_{11}&m_{12}\\
   m_{14}&m_{13}&m_{12}&m_{11}\\
    \end{pmatrix}.
\end{align*}
In the case of the K2P matrices, we additionally have $m_{12} = m_{13}$, while in the case of JC matrices,  $m_{12} = m_{13} = m_{14}$.
It can be easily seen that K3P, K2P, and JC Markov (rate) matrices are centrosymmetric.}

Let us define the following matrix
\begin{align} \label{eq:S}
    {\color{black}S=\begin{pmatrix}
1&0&0&1\\
0&1&1&0\\
0&1&-1&0\\
1&0&0&-1
\end{pmatrix}};
\end{align} 
{\color{black}compare \cite[Section 6]{casanellas2013generating}. For a {\color{black}$4\times 4$ CS} Markov matrix $M$, we define $F(M):=S^{-1}M S$}. By direct computation, it can be checked that $F(M)$ is a block diagonal matrix
\begin{align}\label{Fourier transform of SSM MArkov matrices}
    F(M)=
    \begin{pmatrix}
    \lambda&1-\lambda&0&0\\
    1-\mu&\mu&0&0\\
    0&0&\alpha&\alpha'\\
    0&0&\beta'&\beta\\
    \end{pmatrix},
\end{align}
where

\begin{eqnarray}\label{eq:alpha,beta}
\begin{split}
    \lambda=m_{11}+m_{14},\ & \qquad &\mu=m_{22}+m_{23},\\
\alpha=m_{22}-m_{23},\ & \qquad &{\color{black}\alpha'=m_{21}-m_{24}},\\
\beta=m_{11}-m_{14},\ & \qquad &{\color{black}\beta'=m_{12}-m_{13}}.
\end{split}
\end{eqnarray}
{\color{black}We will then define two matrices,   $M_1:=\begin{pmatrix}\lambda&1-\lambda\\1-\mu&\mu\\\end{pmatrix}$ and $M_2:=\begin{pmatrix}\alpha&\alpha'\\\beta'&\beta\\\end{pmatrix}$, which are the upper and lower block matrices in (\ref{Fourier transform of SSM MArkov matrices}), respectively}.

Similarly, the rate matrices in strand symmetric models are assumed to have the {\color{black} $4 \times 4$ centrosymmetric} form
\begin{align*}Q=
    \begin{pmatrix}
    q_{11}&q_{12}&q_{13}&q_{14}\\
    q_{21}&q_{22}&q_{23}&q_{24}\\
    q_{24}&q_{23}&q_{22}&q_{21}\\
    q_{14}&q_{13}&q_{12}&q_{11}\\
    \end{pmatrix},
\end{align*}
where $$q_{11}+q_{12}+q_{13}+q_{14}=0=q_{21}+q_{22}+q_{23}+q_{24}\mbox{ and }q_{ij}\geq 0\mbox{ for } i\neq j.$$
{\color{black}So}, for a {\color{black} $4\times 4$ CS} rate matrix $Q$, we can also define $F(Q):=S^{-1}Q S$. By direct computation, it can be checked that
\begin{align}\label{eq:F(Q)}
    F(Q)=
    \begin{pmatrix}
    -\rho&\rho&0&0\\
    \sigma&-\sigma&0&0\\
    0&0&\delta&\delta'\\
    0&0&\gamma'&\gamma\\
    \end{pmatrix},
\end{align}
where
\begin{eqnarray*}
\rho=q_{12}+q_{13}, & \sigma=q_{21}+q_{24},\\
\delta=q_{22}-q_{23}, &{\color{black}\delta'=q_{21}-q_{24}},\\
\gamma=q_{11}-q_{14}, &{\color{black}\gamma'=q_{12}-q_{13}}.
\end{eqnarray*}
{\color{black}We will then define two matrices,  $Q_1:=\begin{pmatrix}-\rho&\rho\\\sigma&-\sigma\\\end{pmatrix}$ and $Q_2:=\begin{pmatrix}\delta&\delta'\\\gamma'&\gamma\\\end{pmatrix}$, which are the upper and lower block matrices in (\ref{eq:F(Q)}), respectively}.

The following results provide necessary conditions for a {\color{black}$4\times 4$ CS} Markov matrix to be embeddable.

\begin{lemma}\label{lemma:necessary condition 4 x 4}
Let $M=(m_{ij})$ be a $4\times 4$  CS Markov matrix and $M=\exp(Q)$ for some CS rate matrix $Q$. Then 
\begin{enumerate}
    \item $m_{11}+m_{14}+m_{22}+m_{23}>1$ and
    \item $(m_{22}-m_{23})(m_{11}-m_{14})>(m_{24}-m_{21})(m_{13}-m_{12}).$
\end{enumerate}

\end{lemma}
\begin{proof}
We have that
$$F(M)=S^{-1}MS=S^{-1} \exp{(Q)} S=\exp(S^{-1}QS)=\exp{(F(Q))}.$$
Then 
$$\begin{pmatrix}
M_1&0\\0&M_2\\
\end{pmatrix}
=\exp(F(Q))=\begin{pmatrix}
\exp(Q_1)&0\\0& \exp(Q_2)\\
\end{pmatrix}.$$ Thus, $M_1$ is an embeddable $2\times 2$ Markov matrix.
Using the embeddability criteria of $2\times 2$ Markov matrices in \cite{kingman1962imbedding}, we have that $1<tr(M_1)=\lambda+\mu$, which is the desired inequality. {\color{black}Additionally, since $M_2=\exp(Q_2)$, $det(M_2)>0$ as desired.}
\end{proof}

\begin{lemma}
Let $M=(m_{ij})$ be a $4\times 4$ CS Markov matrix and $M=\exp(Q)$ for some CS rate matrix $Q=(q_{ij})$. {\color{black}If $\lambda+\mu\neq 2$}, then
$$q_{12}+q_{13}=\frac{-\lambda+1}{\lambda+\mu-2}\ln (\lambda+\mu-1)$$
and
$$q_{21}+q_{24}=\frac{-\mu+1}{\lambda+\mu-2}\ln (\lambda+\mu-1).$$
\end{lemma}
\begin{proof}
By direct computations and the proof of Lemma~\ref{lemma:necessary condition 4 x 4},
$$M_1=\exp(Q_1)=\frac{1}{\rho+\sigma}\begin{pmatrix}
e^{-\rho-\sigma}\rho+\sigma&-e^{-\rho-\sigma}\rho+\rho\\
{\color{black}-e^{-\rho-\sigma}\sigma+\sigma}&e^{-\rho-\sigma}\sigma+\rho\\
\end{pmatrix}.$$
We then have the following system of equations:
\begin{equation}\label{lambdamu}
\lambda=\frac{e^{-\rho-\sigma}\rho+\sigma}{\rho+\sigma}\mbox{ and }\mu=\frac{e^{-\rho-\sigma}\sigma+\rho}{\rho+\sigma}.
\end{equation}
Summing the two equations, we get 
$$\lambda+\mu=e^{-\rho-\sigma}+1.$$
Note that by Lemma~\ref{lemma:necessary condition 4 x 4}, $\lambda+\mu>1.$ Therefore, 
\begin{equation}\label{rhosigma}
\rho+\sigma=-\ln (\lambda+\mu-1).
\end{equation}
Using Equation~(\ref{lambdamu}) and (\ref{rhosigma}), we obtain
$$\rho=\frac{-\lambda+1}{\lambda+\mu-2}\ln (\lambda+\mu-1)\qquad\mbox{and}\qquad\sigma=\frac{-\mu+1}{\lambda+\mu-2}\ln (\lambda+\mu-1).$$
{\color{black}The proof is now complete. }
\end{proof}

\smallskip

\smallskip

\begin{proposition}
\label{prop:block diagonal}
Given two matrices $A=(a_{ij}),B=(b_{ij}) \in M_2(\R)$, consider the block-diagonal matrix {\color{black}$C= diag(A,B)$}. Then the following statements hold: 
\begin{itemize}
    \item[i)]  {\color{black}$F^{-1}(C) := S C S^{-1}$} is a CS matrix. 
    \item[ii)]  {\color{black}$F^{-1}(C)$} is a Markov matrix if and only if $A$ is a Markov matrix and
    $$|b_{2 2}|\leq a_{1 1}, \qquad |b_{2 1}|\leq a_{1 2}, \qquad |b_{1 2}|\leq a_{ 2 1}, \qquad |b_{1 1}|\leq a_{2 2}.$$ 

    \item[iii)]  {\color{black}$F^{-1}(C)$} is a rate matrix if and only if $A$ is a rate matrix and 
    $$b_{2 2}\leq a_{1 1}(\leq0), \qquad |b_{2 1}|\leq a_{1 2} (=-a_{1 1}), \qquad |b_{1 2}|\leq a_{2 1}(=-a_{2 2}), \qquad b_{1 1}\leq a_{2 2}(\leq0).$$
    
\end{itemize}

\end{proposition}
\begin{proof}

To prove i), by direct computation we obtain that
 {\color{black}$$F^{-1}(C)= S C S^{-1} = \frac{1}{2}\begin{pmatrix}
        a_{1 1}+b_{2 2}  & a_{1 2}+b_{2 1}  & a_{1 2} - b_{2 1} & a_{1 1} - b_{2 2}\\
        a_{2 1}+b_{1 2}  & a_{2 2} + b_{1 1}  & a_{2 2} - b_{1 1} & a_{2 1} -b_{1 2}\\
        a_{2 1} - b_{1 2} & a_{2 2} - b_{1 1}  & a_{2 2} + b_{1 1} &  a_{2 1}+b_{1 2} \\
         a_{1 1}-b_{2 2}  & a_{1 2} -b_{2 1}  & a_{1 2} +b_{2 1} & a_{1 1} + b_{2 2}\\
    \end{pmatrix}.$$}
Then ii) follows from the above expression of $F^{-1}(Q)$ and the fact that rows of Markov matrices add to 1 and the entries are non-negative, while iii) similarly follows  from the fact that the rows of rate matrices add to zero and the off-diagonal entries are non-negative.{\color{black}}
\end{proof}

For any $4\times 4$ CS Markov matrix $M=(m_{ij})$, let us recall that by (\ref{Fourier transform of SSM MArkov matrices}), $M$ is block-diagonalizable via the matrix $S$. In the rest of this section, we will study both the upper and the lower block matrices of $F(M)$ more closely. {\color{black}Studying the upper and lower blocks allows
us to establish the main result of the embeddability criteria for $4\times 4$
 CS Markov matrices. This block-diagonalization reduces our analysis to studying the logarithms of both the upper and the lower block matrices which have size $2\times 2$. This result will be  presented in Theorem \ref{theorem:criteria SSM-embeddability}.}

\subsubsection*{Upper block}\label{section:upper block}

As we have seen in~(\ref{Fourier transform of SSM MArkov matrices}), the upper block of $F(M)$ is given by the $2\times 2$ matrix $M_1=\begin{pmatrix}
\lambda & 1-\lambda \\
1-\mu&\mu\\
\end{pmatrix}$, which is a Markov matrix. If $P_1=\begin{pmatrix}
1&1-\lambda\\
1&\mu-1\\
\end{pmatrix}$, then 
$$P_1^{-1}M_1P_1=\begin{pmatrix}
1&0\\
0&\lambda+\mu-1\\
\end{pmatrix}.$$ Hence, by Proposition~\ref{Prop:Logenum}, any  logarithm of $M_1$ can be written as
$$L^{M_1}_{k_1,k_2}:=P_1\begin{pmatrix}
2k_1\pi i&0\\
0&\log(\lambda+\mu-1)+2k_2\pi i\\
\end{pmatrix}P_1^{-1},$$
for some integers $k_1$ and $k_2$.
Let $p=\log(\lambda+\mu-1)$, $q=1-\lambda$, and $r=1-\mu$. Then
\begin{eqnarray}\label{eqn:upperblock}
{\color{black}L^{M_1}_{k_1,k_2}}=\frac{1}{2-\lambda-\mu}\begin{pmatrix}
qp+2\pi(rk_1+qk_2)i&{\color{black}-qp+2\pi q(k_1-k_2)i}\\
{\color{black}-rp+2\pi r(k_1-k_2)i}&rp+2\pi(qk_1+rk_2)i\\
\end{pmatrix}.
\end{eqnarray}

\begin{lemma}\label{lemma:real logarithm of upper block}
If $\lambda+\mu\neq 2$, then $L^{M_1}_{k_1,k_2}$ is a real matrix if and only if $k_1=k_2=0$ and $\lambda+\mu>1$. In this case, the only real logarithm of $M_1$ is the principal logarithm $$\frac{1}{2-\lambda-\mu}\begin{pmatrix}
qp&-qp\\
-rp&rp\\
\end{pmatrix}.$$
\end{lemma}
\begin{proof}
For fixed $k_1$ and $k_2$, the eigenvalues of $L^{M_1}_{k_1,k_2}$ are $\lambda_1=2k_1\pi i$ and $\lambda_2=p+2k_2\pi i.$ Then $L^{M_1}_{k_1,k_2}$ is a real matrix if and only if $\lambda_1,\lambda_2\in \mathbb{R}$ or $\lambda_2=\overline{\lambda_1}$. Since $\lambda+\mu\neq 2$, $\lambda_2\neq\overline{\lambda_1}$. Thus, $L^{M_1}_{k_1,k_2}$ is a real matrix if and only if $\lambda_1,\lambda_2\in \mathbb{R}$. {\color{black}Finally, $\lambda_1\in \mathbb{R}$ if and only if $k_1=0$ and $\lambda_2\in \mathbb{R}$ if and only if $k_2=0$ and $\lambda+\mu>1$.} {\color{black}}
\end{proof}

\subsubsection*{Lower block}\label{section:lower block}

The lower block of $F(M)$ is given by the matrix $M_2=\begin{pmatrix}
\alpha&\alpha'\\\beta'&\beta\\
\end{pmatrix}$. Unlike $M_1$, the matrix $M_2$ is generally not a Markov matrix. The discriminant\index{discriminant $\Delta$ of characteristic polynomial of $M_2$} of the characteristic polynomial of $M_2$ is given by 
\begin{equation}
\label{eq:discriminant}
    \Delta:=(\alpha-\beta)^2+4\alpha'\beta'
\end{equation} 
{\color{black} with  $\alpha,\beta,\alpha',\beta'$ defined as in \eqref{eq:alpha,beta}.} If $\Delta>0$, then $M_2$ has two distinct real eigenvalues and if $\Delta<0$, then $M_2$ has a pair of conjugated complex eigenvalues. Moreover, if $\Delta=0$, then $M_2$ has either $2\times 2$ Jordan block or a repeated real eigenvalue. We will assume that $\Delta \neq 0$ so that $M_2$ diagonalizes into two distinct eigenvalues.

Let $P_2=\begin{pmatrix}
\frac{\sqrt{\Delta}+(\alpha-\beta)}{2}&\frac{\sqrt{\Delta}-(\alpha-\beta)}{2}\\
\beta'&{\color{black}-\beta'}\\
\end{pmatrix}$. Then 
$$P_2^{-1}M_2P_2=\begin{pmatrix}
\frac{(\alpha+\beta)+\sqrt{\Delta}}{2}&0\\
0&\frac{(\alpha+\beta)-\sqrt{\Delta}}{2}\\
\end{pmatrix}.$$
Let us now define 
$$l_3:=\log(\frac{(\alpha+\beta)+\sqrt{\Delta}}{2})+2k_3\pi i\qquad \mbox{and}\qquad l_4:=\log(\frac{(\alpha+\beta)-\sqrt{\Delta}}{2})+2k_4\pi i,$$
where $k_3$ and $k_4$ are integers. Therefore, any logarithm of $M_2$ can be written as 

{\color{black}\begin{eqnarray}\label{eqn:lowerblock}
L^{M_2}_{k_3,k_4}:=\begin{pmatrix}
\varepsilon&\phi\\
\gamma&\eta\\
\end{pmatrix} 
\end{eqnarray}}
where
\begin{align*}
    \varepsilon & :=\frac{1}{2}((l_3+l_4)+(\alpha-\beta)\frac{(l_3-l_4)}{\sqrt{\Delta}}),\\
    \phi& :=\alpha'\frac{(l_3-l_4)}{\sqrt{\Delta}},\\
    \gamma&:=\beta'\frac{(l_3-l_4)}{\sqrt{\Delta}}\mbox{ and }\\
    \eta&:=\frac{1}{2}((l_3+l_4)-(\alpha-\beta)\frac{(l_3-l_4)}{\sqrt{\Delta}}).
\end{align*}
\begin{lemma}\label{lemma:real logarithm of lower block}
\begin{enumerate}
    \item If $\Delta>0$, then $L^{M_2}_{k_3,k_4}$ is a real matrix if and only if $\alpha+\beta>\sqrt{\Delta}$ and $k_3=k_4=0.$
    \item If $\Delta<0$, then $L^{M_2}_{k_3,k_4}$ is a real matrix if and only if $k_4=-k_3$. 
\end{enumerate}
\end{lemma}
\begin{proof}
\begin{enumerate}
    \item If $\Delta>0$, then $Im(l_3)=2k_3\pi$ and $Im(l_4)=2k_4\pi$. Moreover, $Re(l_3)\neq Re(l_4).$ Since $l_3$ and $l_4$ are the eigenvalues of $L^{M_2}_{k_3,k_4}$, this implies that $l_3\neq \overline{l}_4$. In particular, $L^{M_2}_{k_3,k_4}$ is a real matrix if and only if both $l_3$ and $l_4$ are real.
    \item Let us assume $\Delta<0$ and take $z=\frac{(\alpha+\beta)+\sqrt{\Delta}}{2}$. Fixing $k_3,k_4\in\mathbb{Z}$, the eigenvalues of $L^{M_2}_{k_3,k_4}$ are $l_3=\log(z)+2k_3\pi i$ and $l_4=\log(\overline{z})+2k_4\pi i=\overline{Log(z)}+2k_4\pi i$, which are both complex numbers. Thus, $L^{M_2}_{k_3,k_4}$ is real if and only if $l_3=\overline{l}_4$. Hence, $k_4=-k_3.$ Conversely, $k_4=-k_3$ implies that
    $l_3+l_4=2 Re(l_3)\in\mathbb{R}$ and $\frac{l_3-l_4}{\sqrt{\Delta}}=\frac{2 Im(l_3)i}{\sqrt{\Delta}}\in \mathbb{R}$. Thus, all entries of  $L^{M_2}_{k_3,k_4}$ are real.{\color{black}}
\end{enumerate}
\end{proof}

\subsubsection*{Logarithms of $4\times 4$ CS Markov matrices}
 Let $M$ be a $4\times 4$ CS Markov matrix.  {\color{black} Using the values defined in \eqref{eq:alpha,beta} and \eqref{eq:discriminant},} we can now label up its four eigenvalues, namely, 
 \begin{equation}\label{eq:eigenvalues}
    1,\quad \lambda_1:=\lambda+\mu-1,\quad {\color{black}\lambda_2}:=\frac{(\alpha+\beta)+\sqrt{\Delta}}{2}\quad \mbox{ and }\quad{\color{black}\lambda_3}=\frac{(\alpha+\beta)-\sqrt{\Delta}}{2}. 
 \end{equation}

 We note that the subset of $4\times 4$ CS Markov matrix with repeated eigenvalues (diagonalizing matrix with repeated eigenvalues or a Jordan block of size greater than 1) have zero measure. {\color{black}  Therefore generic $4\times4$ Markov matrices have no repeated eigenvalues}, and hence we are going to assume the eigenvalues to be {\color{black} distinct} . In particular, we are assuming that $M$ diagonalizes. Furthermore, since we want $M$ to have real logarithms and have no repeated eigenvalues, we need the real eigenvalues to be positive.
 
 The following theorem characterizes the embeddability of a $4\times 4$ CS Markov matrix {\color{black}with positive and distinct eigenvalues}. Furthermore, the theorem guarantees that a $4\times 4$ CS Markov matrix is embeddable {\color{black}if and only if} it admits a CS Markov generator. In particular, the characterization of the embeddability of a CS matrix is equivalent when restricting to rate matrices satisfying the symmetries imposed by the model (model embeddability) than when restricting to all possible rate matrices (embedding problem).

\begin{theorem}\label{thm:embeddability SSM}
 Let $M$ be a diagonalizable $4\times 4$ CS Markov matrix with {\color{black}positive} and {\color{black} distinct} eigenvalues {\color{black}$\lambda_1, \lambda_2, \lambda_3$ defined as in \eqref{eq:eigenvalues}}. Let us define
 $$x=\log(\lambda_1),\qquad y_k=\log({\color{black}\lambda_2})+2k\pi i,\qquad z_k=\log({\color{black}\lambda_3})-2k\pi i,$$
 where $k=0$ if $\Delta>0$ and $k\in \mathbb{Z}$ if $\Delta<0.$ Then any real logarithm of $M$ is given by
 $$S\begin{pmatrix}
 \alpha_1&-\alpha_1&0&0\\
 -\beta_1&\beta_1&0&0\\
 0&0&\delta(k)&\varepsilon(k)\\
 0&0&\phi(k)&\gamma(k)\\
 \end{pmatrix}S^{-1},$$
 where 
 \begin{align*}
     \alpha_1&= \frac{1-\lambda}{2-\lambda-\mu}x,\qquad &\beta_1&=\frac{1-\mu}{2-\lambda-\mu}x,\\
     \delta(k)&=\frac{1}{2}(({\color{black}y_k+z_k})+(\alpha-\beta)\frac{({\color{black}y_k-z_k})}{\sqrt{\Delta}}),\qquad &\varepsilon(k)&=\alpha'\frac{({\color{black}y_k-z_k})}{\sqrt{\Delta}},\\
     \phi(k)&=\beta'\frac{({\color{black}y_k-z_k})}{\sqrt{\Delta}},\qquad &\gamma(k)&=\frac{1}{2}(({\color{black}y_k+z_k})-(\alpha-\beta)\frac{({\color{black}y_k-z_k})}{\sqrt{\Delta}}).
 \end{align*}
{\color{black} with $\lambda,\ \mu,\ \alpha,\ \beta,\ \alpha' \text{ and } \beta` $ defined as in \eqref{eq:alpha,beta} and $\Delta$ as in \eqref{eq:discriminant}.}\\ 
 
In particular, any real logarithm of $M$ is also a $4\times 4$ CS matrix {\color{black}whose entries} $q_{11},\dots,q_{24}$ are given by:
 \begin{align*}
     &q_{11}=\frac{\alpha_1+\gamma(k)}{2},\quad  &q_{12}=\frac{-\alpha_1+\phi(k)}{2},\quad
      &q_{13}=\frac{-\alpha_1-\phi(k)}{2},\quad &q_{14}=\frac{\alpha_1-\gamma(k)}{2},\\
      &q_{21}=\frac{-\beta_1+\varepsilon(k)}{2},\quad
      &q_{22}=\frac{\beta_1+\delta(k)}{2},\quad
      &q_{23}=\frac{\beta_1-\delta(k)}{2},\quad
      &q_{24}=\frac{-\beta_1-\varepsilon(k)}{2}.\\
 \end{align*}
\end{theorem}
\begin{proof}
Let us note that 
$$M=S\cdot \mbox{diag}(P_1,P_2)\cdot \mbox{diag}(1, \lambda_1, {\color{black}\lambda_2, \lambda_3})\cdot  \mbox{diag}(P_1^{-1},P_2^{-1})\cdot S^{-1}.$$
Since we assume that the eigenvalues of $M$ are {\color{black}distinct}, {\color{black}according to Proposition~\ref{Prop:Logenum}}, any logarithm of $M$ can be written as
\begin{align*}
   Q&=S\cdot \mbox{diag}(P_1,P_2)\cdot \mbox{diag}(\log_{k_1}(1), \log_{k_2}(\lambda_1),\log_{k_3}({\color{black}\lambda_2}),\log_{k_4}({\color{black}\lambda_3}))\cdot \mbox{diag}(P_1^{-1},P_2^{-1})\cdot S^{-1}\\
    &=S\cdot \mbox{diag}(L^{M_1}_{k_1,k_2},L^{M_2}_{k_3,k_4})\cdot S^{-1},
\end{align*}

{\color{black}  The last equation and the fact that $S$ and $S^{-1}$ are real matrices imply that} $Q$ will be real if and only if both $L^{M_1}_{k_1,k_2}$ and $L^{M_2}_{k_3,k_4}$ are real. {\color{black} Here $L^{M_1}_{k_1,k_2}$ is the upper block given in~(\ref{eqn:upperblock}) and $L^{M_2}_{k_3,k_4}$ is the lower block defined in~(\ref{eqn:lowerblock})}.
{\color{black} By Lemma~\ref{lemma:real logarithm of upper block}, $L^{M_1}_{k_1,k_2}$ being a real logarithm implies that $k_1=k_2=0$ and $\lambda+\mu>1$. Then $L^{M_2}_{k_3,k_4}$ being a real matrix, according to Lemma~\ref{lemma:real logarithm of lower block}, implies that $k_3=k_4=0$ if $\Delta>0$, while $k_4=-k_3$ if $\Delta<0.$ Therefore, the upper block is $L^{M_1}_{0,0}$ and the lower block will be $L^{M_2}_{k,-k},\text{ for } k=k_3$ completing the proof. }
{\color{black}}
\end{proof}

 {\color{black}Now we are interested in knowing when the real logarithm of a $4\times 4$ CS Markov matrix is a rate matrix. Using the same notation as in Theorem \ref{thm:embeddability SSM} we get the following result.}

 \begin{theorem}\label{theorem:criteria SSM-embeddability}
  A diagonalizable $4\times 4$ CS Markov matrix $M$ with {\color{black}distinct}  eigenvalues is embeddable if and only if the following conditions hold for $k=0$ if $\Delta>0$ or for some $k\in\mathbb{Z}$ if $\Delta<0$:
  \begin{align*}
      \lambda_1>0,\qquad (\alpha+\beta)^2>\Delta,\qquad  |\phi(k)|\leq -\alpha_1,\qquad |\varepsilon(k)|\leq -\beta_1,\qquad \gamma(k)\leq \alpha_1,\qquad \delta(k)\leq \beta_1.
  \end{align*}
 \end{theorem}
 
 {\color{black}\begin{proof}
 The logarithm of a $4\times 4$ CS Markov matrix will depend on whether $\Delta>0$ or $\Delta<0$. In particular, it will depend on whether the eigenvalues $\lambda_2$ and $\lambda_3$ are real and positive or whether they are conjugated complex numbers. 
 \begin{enumerate}
    \item If $\Delta>0$, then both $\lambda_2$ and $\lambda_3$ are real and $\lambda_2>\lambda_3$. Hence, $z<y<0.$ Moreover, Lemma~\ref{lemma:real logarithm of lower block} implies that $\lambda_3>0$ and hence $\lambda_2\lambda_3>0$.
    \item If $\Delta<0$, then $\lambda_2,\lambda_3\in\mathbb{C}\setminus\mathbb{R}$ and $\lambda_2=\overline{\lambda_3}$. Hence, $y+z>0$ and $y-z=4\pi ki.$ Moreover, $\lambda_2\lambda_3=|\lambda_3|^2>0$ since $\lambda_3\neq 0$.
\end{enumerate}
 Thus, in both cases, $\alpha_1,\beta_1,\delta(k),\varepsilon(k), \phi(k),\gamma(k)\in \mathbb{R}.$ Moreover, $\alpha_1$ and $\beta_1$ are both non-positive. In particular, Theorem~\ref{thm:embeddability SSM} together with Proposition~\ref{prop:block diagonal} imply that a real logarithm of $M$ is a rate matrix if and only if 
$$|\phi(k)|\leq -\alpha_1,\qquad |\varepsilon(k)|\leq -\beta_1,\qquad \gamma(k)\leq \alpha_1,\qquad \delta(k)\leq \beta_1.$$
Furthermore, the conditions $\lambda_1>0$  comes from  Lemma~\ref{lemma:real logarithm of upper block}. The proof is now complete.
 \end{proof}}

{\color{black}
\begin{remark}\label{rk:k_bound}
According to Theorem \ref{theorem:criteria SSM-embeddability} the embeddability of a $4\times4$ CS Markov matrix $M$ with distinct positive eigenvalues can be decided by checking six inequalities depending on the entries of $M$. However, if $M$ has non-real eigenvalues then one has to check infinitely many groups of inequalities, one for each value of $k\in \mathbb{Z}$. It is enough that one of those systems is consistent to guarantee that $M$ is embeddable. Theorem 5.5 in \cite{casanellas2020embedding} provides boundaries for the values of $k$ for which the corresponding inequalities may hold.
\end{remark}
}

Let us take a look at the class of K3P matrices which is a special case of strand symmetric matrices. Indeed, for a K3P matrix $M=(m_{ij})$, we have that 
$$m_{11}=m_{22},\qquad m_{12}=m_{21},\qquad m_{13}=m_{24}\qquad \mbox{and}\qquad m_{14}=m_{23}.$$
Suppose that a K3P-Markov matrix $M=(m_{ij})$ is K3P-embeddable, i.e. $M=\exp(Q)$
for some K3P-rate matrix $Q.$ Recall that the eigenvalues of $M$ are 
 $$1, \quad p:=m_{11}+m_{12}-m_{13}-m_{14},\quad q:=m_{11}-m_{12}+m_{13}-m_{14}\quad\mbox{ and }\quad r:=m_{11}-m_{12}-m_{13}+m_{14}.$$
 In this case, we have that
 \begin{align*}
    &\lambda=\mu=m_{11}+m_{14},\quad
     \alpha=\beta=m_{11}-m_{14},\quad
     \alpha'=\beta'=m_{13}-m_{12},\quad\\
     &\lambda_1=r\quad
     \mbox{ and }\quad\Delta=4(m_{13}-m_{12})^2.\\
 \end{align*}
 In particular, we see that $\Delta> 0$ unless $m_{12}=m_{13}$. Moreover,
 \begin{align*}
      &x=\log r,\quad y=\log q,\quad z=\log p,\quad \alpha_1=\beta_1=\frac{1}{2}\log r,\\
      &\quad\delta(0)=\gamma(0)=\frac{1}{2}\log pq,\quad |\varepsilon(0)|=|\phi(0)|=\frac{1}{2}\log \frac{q}{p}.\\
 \end{align*}
 The inequalities in Theorem~\ref{theorem:criteria SSM-embeddability} can be spelled out as follows:
 $$r>0,\quad pq>0,\quad |\log \frac{q}{p}|\leq -\log r\quad \mbox{and}\quad \log pq\leq \log r.$$
These inequalities are equivalent to the K3P-embeddability criteria presented in \cite[Theorem 3.1]{roca2018embeddability} and \cite[Theorem 1]{ardiyansyah2021model}. {\color{black} Moreover, they are also equivalent to the restriction to centrosymmetric-matrices of the embeddability criteria for $4\times4$ Markov matrices with different eigenvalues given in \cite[Theorem 1.1]{casanellas2020embedding} }

In the last part of this section, we discuss the rate identifiability problem for $4\times 4$ centrosymmetric matrices. If a centrosymmetric Markov matrix arises from a continuous-time model, then we want {\color{black}to determine its corresponding substitution rates}. Namely, given an embeddable $4\times 4$ CS matrix, we want to know if we can uniquely identify its Markov generator. 

{ \color{black} 
It is worth noting that Markov matrices with repeated real eigenvalues may admit more than one Markov generator (e.g. examples 4.2 and 4.3 in \cite{casanellas2020embeddability} show embeddable K2P matrices with more than one Markov generator). Nonetheless, this is not possible if the Markov matrix has distinct eigenvalues, because in this case its only possible real logarithm would be the principal logarithm \cite{Culver}. 
As one considers less restrictions in a model, the measure of the set of matrices with repeated real eigenvalues decreases, eventually becoming a measure zero set. For example, this is the case within the K3P model, where both its submodels (the K2P model and the JC model) consist of matrices with repeated eigenvalues and have positive measure subsets of embeddable matrices with non-identifiable rates. However, when considering the whole set of K3P Markov matrices, the subset of embeddable matrices with more than one Markov generator has measure zero (see Chapter 4 in \cite{Roca}). Nevertheless, this behaviour only holds if the Markov matrices within the model have real eigenvalues.  }
\begin{proposition}
\label{prop:rate_idendifiable}
There is a positive measure subset of $4\times 4$ CS Markov matrices that are embeddable and whose rates are not identifiable. Moreover, all the Markov generators of the matrices in this set are also CS matrices.
\end{proposition}
\begin{proof}
Given  \begin{center}
    
$P= \begin{small}  \begin{pmatrix}
    1 & -5 & 1-i & 1+i\\
    1 & 2 & -i & i\\
    1 & 2 & i & -i\\
    1 & -5 & -1+i & -1-i\\
    \end{pmatrix} \end{small},$\end{center}
let us consider the following matrices
   $$ M= P \;diag(1, e^{-7\pi}, e^{-4\pi }i, - e^{-4\pi }i)\; P^{-1}, \qquad Q=  P\; diag(0, -7\pi, -4\pi -\frac{3\pi}{2}i,  -4\pi +\frac{3\pi}{2}i)\; P^{-1}. $$
A straightforward computation shows that  $M$ is a CS Markov matrix and $Q$ is a CS rate matrix. Moreover they both have non-zero entries. By applying the exponential series to $Q$, we get that $\exp(Q)=M$. That is $M$ is embeddable and $Q$ is a Markov generator of $M$.

Since $Q$ is a rate matrix, so is $Qt$ for any $t\in \mathbb{R}_{\geq0}$.  {\color{black}Therefore, $\exp(Qt)$ is an embeddable Markov matrix, because the exponential of any rate matrix is necessarily a Markov matrix. See \cite[Theorem 4.19]{pachter2005algebraic} for more details.} Moreover, we have that
$${\color{black}S^{-1}P = \begin{pmatrix}
1 & -5 & 0 &0 \\
1 & 2 &0 & 0\\
 0&0 & -i & i\\
 0 & 0& 1-i & 1+i\\
\end{pmatrix}},$$
so $S^{-1} \exp(Qt)S $  is a 2-block diagonal matrix. Hence, {\color{black}by Proposition~\ref{prop:block diagonal}} we have that $\exp(Qt)$  is an embeddable strand symmetric Markov matrix for all $t\in \mathbb{R}_{>0}$.

Now, let us define $V=P\;diag(0,0,2\pi {\color{black}i}, -2\pi {\color{black}i})\;P^{-1}$. Note that $Q$ and $V$ diagonalize simultaneously via $P$ and hence they commute. Therefore, $$\exp(Q+V)= \exp(Q)\exp(V)=M I_4= M$$ by the Baker-Campbell-Haussdorff formula. Moreover,  $$\exp(Qt+kV)= \exp(Qt)\exp(kV)=\exp(Qt)I_4= \exp(Qt)$$ for all $k \in \mathbb{Z}$. Note that $kV$ is a bounded matrix for any given $k$ and hence, given $t$ large enough, it holds that $Qt+mV$ is a rate matrix for any $m$ between $0$ and $k$.

This shows that, for $t$ large enough, $\exp(Qt)$ is an embeddable CS Markov matrix with at least $k+1$ different CS Markov generators. Moreover, $\exp(Qt)$ and all its generators have no null entries by construction and they can therefore be perturbed as  in Theorem 3.3 in  \cite{casanellas2020open}
 to obtain a positive measure subset of embeddable CS Markov matrices that have $k+1$ CS Markov generators. {\color{black} Such perturbation consists of small enough changes on the real and complex parts of the eigenvalues and eigenvectors of $M$ (other than the  eigenvector $(1,\dots,1)$ and its corresponding eigenvalue $1$.) } 
\end{proof}

\begin{remark}
Using the same notation as in the proposition above and given $C\in GL_2(\mathbb{C})$, let us define
 $$Q(C) = P \;diag(I_2,C)\;diag\left(1, -7\pi, -4\pi -\frac{3\pi}{2}i,  -4\pi +\frac{3\pi}{2}i\right)\;diag(I_2,C^{-1})\;P^{-1}.$$ Since $Q(I_2)=Q$ is a CS rate matrix with no null entries, so is $Q(C)$ for $C\in GL_2(\mathbb{C})$ close enough to $I_2$. Moreover, by construction we have that  $\exp(2tQ(C))=\exp(2tQ)$ for all $t\in \mathbb{N}$. Therefore, for $t\in \mathbb{N}$ we have that $\exp(2tQ)$ has   {\color{black}uncountably} many Markov generators (namely, $2tQ(C)$ with $C$ close to $I_2$) and all of them are CS matrices \cite[Corollary 1]{Culver}. It is worth noting that {\color{black}according to \cite[Corollary 1]{Culver},} if a matrix has {\color{black}uncountably} many logarithms, then it necessarily has repeated real eigenvalues. Therefore, the subset of embeddable CS Markov matrices with {\color{black}uncountably} many generators has measure zero within the set of all  matrices.

\end{remark}

\section{Volumes of $4\times 4$ CS Markov matrices}
\label{sect:volume}

In this section, we compute the relative volumes of embeddable {\color{black}$4\times 4$ CS} Markov matrices within some meaningful subsets of Markov matrices. The aim of this section is to describe how large the different sets of matrices are compared to each other.

Let $V^{Markov}_4$ be the set of all $4\times 4$ {\color{black} CS} Markov matrices\index{$V^{Markov}_4$ set of $4\times 4$ CS Markov matrices}. We use the following description
$$V^{Markov}_4=\{(b,c,d,e,g,h)^T\in \mathbb{R}^6 : b,c,d,e,g,h\geq 0,\quad 1-b-c-d\geq 0,\quad 1-e-g-h\geq 0\}.$$
More explicitly, we identify the $4\times 4$ {\color{black} CS} Markov matrix 
$$\begin{pmatrix}
1-b-c-d&b&c&d\\
e&1-e-g-h&g&h\\
h&g&1-e-g-h&e\\
d&c&b&1-b-c-d\\
\end{pmatrix}$$ with a point $(b,c,d,e,g,h)\in V^{Markov}_4.$
Let $V_+$ \index{$V_+$ set of CS Markov matrices with real positive
eigenvalues} be the set of all {\color{black}CS} Markov matrices having {\color{black} real positive eigenvalues}, where
$$\Delta=((1-e-2g-h)-(1-b-c-2d))^2+4(e-h)(b-c), $$
{\color{black} is the discriminant of the matrix $M_2$} as stated in Section~\ref{section:lower block}. {\color{black}We have $V_+\subseteq V^{Markov}_4$.}
More explicitly,
\begin{small}
\begin{align*}
V_+=\{&(b,c,d,e,g,h)\in \mathbb{R}^6 : b,c,d,e,g,h\geq 0,\qquad  1-b-c-d\geq 0,\qquad 1-e-g-h\geq 0,\qquad 1-b-c-e-h>0,\\&  {\color{black}(2-b-c-2d-e-2g-h)+\Delta>0,\qquad  (2-b-c-2d-e-2g-h)-\Delta>0},\qquad \Delta>0\}.\\
\end{align*}
\end{small}
Let $V_{em+}$ be the set\index{$V_{em+}$ set of embeddable CS Markov matrices with four distinct real positive eigenvalues} of all embeddable {\color{black}$4\times 4$ CS} Markov matrices with {\color{black}four distinct real positive eigenvalues}. {\color{black}We have $V_{em+}\subseteq V_+$.} Therefore, by Theorem~\ref{theorem:criteria SSM-embeddability},
\begin{small}
\begin{align*}
V_{em+}=\{&(b,c,d,e,g,h)\in \mathbb{R}^6 :  b,c,d,e,g,h\geq 0,\qquad 1-b-c-d\geq 0,\qquad  1-e-g-h\geq 0,\qquad 1-b-c-e-h>0,\\&{\color{black}(2-b-c-2d-e-2g-h)+\Delta>0,\qquad  (2-b-c-2d-e-2g-h)-\Delta>0} , \qquad \Delta>0,\\
&|\phi(0)|\leq -\alpha_{{\color{black}1}},\qquad|\varepsilon(0)|\leq -\beta_{{\color{black}1}},\qquad\delta(0)\leq \beta_{{\color{black}1}},\qquad\gamma(0)\leq \alpha_{{\color{black}1}} \}.\\
\end{align*}
\end{small}

Finally, we consider the following two biologically relevant subsets of $V^{Markov}_4$. Let $V_{\rm{DLC}}$\index{$V_{\rm{DLC}}$ set of diagonally largest in column (DLC) Markov matrices} be the set of \textit{diagonally largest in column}\index{diagonally largest in column (DLC) matrix} {\color{black}(DLC) Markov} matrices, which is the subset of $V^{Markov}_4$ containing all {\color{black}CS} Markov matrices such that the diagonal element is the largest element in each column. These matrices are related to matrix parameter identifiability in phylogenetics \cite{chang1996full}. Secondly, we let $V_{\rm{DD}}$\index{$V_{\rm{DD}}$ set of diagonally dominant Markov matrices} be the set of \textit{diagonally dominant\index{diagonally dominant (DD) matrix}} {\color{black}(DD) Markov} matrices, which {\color{black}is} the subset of $V^{Markov}_4$ matrices containing all {\color{black}CS} Markov matrices such that in each row the diagonal element is at least the sum of all the other elements in the row. Biologically, the subspace $V_{\rm{DD}}$ consists of matrices with probability of not mutating at least as large as the probability of mutating. If a diagonally dominant matrix is embeddable, it has an identifiable rate matrix \cite{cuthbert1972uniqueness,Cuthbert}. By the definition of each set, we have the inclusion $V_{\rm{DD}}\subseteq V_{\rm{DLC}}$.

{\color{black}
\begin{remark}
The sets $V_+$, $V_{em+}$, $V_{\rm{DLC}}$, $V_{\rm{DD}}$ that we consider in this section are all subsets of the set $V^{Markov}_4$ of all $4\times 4$ {\color{black} CS} Markov matrices, but we can use the same definition to refer to the equivalent subsets of $n\times n$ {\color{black} CS} Markov matrices. Therefore, we will use the same notation $V_+$, $V_{em+}$, $V_{\rm{DLC}}$, $V_{\rm{DD}}$ to refer to the equivalent subsets of the set $V^{Markov}_n$ of $n\times n$ {\color{black} CS} Markov matrices without confusion in the following sections.
\end{remark}
}

In the rest of this section, the number $v(A)$ denotes the Euclidean volume of the set $A$. By definition, $V^{Markov}_4$, $V_{\rm{DLC}}$ and $V_{\rm{DD}}$ are polytopes, since they are defined by the linear inequalities in $\mathbb{R}^6$. Hence, we can use \texttt{Polymake} \cite{gawrilow2000polymake} to compute their exact volumes and obtain that
$$v(V^{Markov}_4)=\frac{1}{36},\quad v(V_{\rm{DLC}})= \frac{1}{576}\quad\mbox{ and }\quad v(V_{\rm{DD}})=\frac{1}{2304}.$$ Hence, we see that $V_{\rm{DLC}}$ and $V_{\rm{DD}}$ constitute roughly only $6.25\%$ and $1.56\%$ of $V_4^{Markov}$, respectively.

On the other hand, we will estimate the volume of the sets $V_+, V_{em+}, V_{\rm{DLC}}\cap V_+,V_{\rm{DLC}}\cap V_{em+},V_{\rm{DD}}\cap V_+,\mbox{ and }V_{\rm{DD}}\cap V_{em+}$ using the hit-and-miss Monte Carlo integration method \cite{hammersley2013monte} with sufficiently many sample points in \texttt{Mathematica} \cite{Mathematica}. {\color{black}Theoretically, Theorem~\ref{theorem:criteria SSM-embeddability} enables us to compute the exact volume of these relevant sets. For example in the case of K3P matrices, such exact computation of volumes has been feasible in \cite{roca2018embeddability}. However, while for the K3P matrices, the embeddability criterion is given by three quadratic polynomial inequalities, in the case of CS matrices the presence of nonlinear and nonpolynomial constraints imposed on each set, makes the exact computation of the volume of these sets intractable. Therefore, we need to approximate the volume of these sets.} Given a subset $A\subseteq V_4^{Markov}$, the volume estimate of $v(A)$ computed using the hit-and-miss Monte Carlo integration method with $n$ sample points {\color{black} is given by the number of points belonging to $A$ out of $n$ sample points}.
For computational purposes, in the formula of $\phi(0)$ and $\varepsilon(0)$, we use the fact that{\color{black}
\begin{align*}
    y-z&=\log\left(\frac{(2-b-c-2d-e-2g-h)+\sqrt{\Delta}}{(2-b-c-2d-e-2g-h)-\sqrt{\Delta}}\right).\\
    &=\log\left(\frac{((2-b-c-2d-e-2g-h)+\sqrt{\Delta})^2}{(2-b-c-2d-e-2g-h)^2-\Delta}\right).\\
    &=\log\left(\frac{((2-b-c-2d-e-2g-h)+\sqrt{(b+c+2d-e-2g-h)^2+4(e-h)(b-c)})^2}{(2-b-c-2d-e-2g-h)^2-((b+c+2d-e-2g-h)^2+4(e-h)(b-c))}\right)
\end{align*}}
All codes for the computations implemented \texttt{Mathematica} and \texttt{Polymake} can be found at the following address:
\texttt{https://github.com/ardiyam1/Embeddability-and-rate-identifiability-of-\\centrosymmetric-matrices}.

The results of these estimations using the hit-and-miss Monte Carlo integration {\color{black}implemented in \texttt{Mathematica}} with $n$ sample points are presented in Table~\ref{table:estimated-volume-M_+ with Delta>0}, while Table~\ref{table:estimated-ratio-volume with Delta>0} provides an estimated volume ratio between relevant subsets of centrosymmetric Markov matrices using again the hit-and-miss Monte Carlo integration with $n$ sample points. {\color{black}In Table~\ref{table:estimated-volume-M_+ with Delta>0}, we firstly generate $n$ centrosymmetric matrices whose off-diagonal entries were sampled uniformly in $[0, 1]$ and forced the rows of the matrix to sum to one. Out of these $n$ matrices, we test how many of them are actually Markov matrices (i.e. the diagonal entries are non-negative) and then out of these how many have positive eigenvalues.}. In particular, for {\color{black}$n=$}$10^7$ sample points {\color{black}containing 277628 centrosymmetric Markov matrices}, Table~\ref{table:estimated-ratio-volume with Delta>0} suggests that there are approximately {\color{black}$1.7\%$ of centrosymmetric Markov matrices with distinct positive eigenvalues that are embeddable. Moreover, we can see that for $n=10^7$, out of all embeddable centrosymmetric Markov matrices with distinct positive eigenvalues, almost all are diagonally largest in column, while only $28\%$ are diagonally dominant.}

\begin{table}[H]
\centering
\caption{{\color{black}Number of samples in the sets } $V_+, V_{em+}, V_{\rm{DLC}}\cap V_+,V_{\rm{DLC}}\cap V_{em+},V_{\rm{DD}}\cap V_+\mbox{ and }V_{\rm{DD}}\cap V_{em+}$ {\color{black}using hit-and-miss methods} and Theorem~\ref{theorem:criteria SSM-embeddability}.}
\label{table:estimated-volume-M_+ with Delta>0}
\begin{tabular}{|c|c|c|c|c|}
    \hline
    $n$  &$10^4$ & $10^5$ & $10^6$ & $10^7$\\
    \hline
    {\color{black}Samples in $V_4^{Markov}$} & {\color{black}280} & {\color{black}2767}& {\color{black}27829} & {\color{black}277628}\\
     \hline
   {\color{black}Samples $V_+$} & {\color{black}23}&{\color{black}192}&{\color{black}1999}&{\color{black}20601}\\
   \hline 
    {\color{black}Samples in $V_{em+}$} & {\color{black}3}&{\color{black}34}&{\color{black}359}&{\color{black}3511}\\
    \hline
    {\color{black}Samples in $V_{\rm{DLC}}\cap V_+$}& {\color{black}19}& {\color{black}154} & {\color{black}1541}& {\color{black}15830}\\\hline
    {\color{black}Samples in $V_{\rm{DD}}\cap V_+$}&{\color{black}3}&{\color{black}31}&{\color{black}262}&{\color{black}2889}\\\hline
    {\color{black}Samples in $V_{\rm{DLC}}\cap V_{em+}$}& {\color{black}3}&{\color{black}34}& {\color{black}357}& {\color{black}3503}\\\hline
    {\color{black}Samples in $V_{\rm{DD}}\cap V_{em+}$}&{\color{black}1}&{\color{black}15}&{\color{black}105}&{\color{black}1011}\\
   \hline
\end{tabular}
\end{table}

\begin{table}[H]
\centering
\caption{Relative volumes ratio between the relevant subsets {\color{black} obtained using hit-and-miss method and Theorem \ref{theorem:criteria SSM-embeddability}. The volumes were estimated as the quotient of the sample sizes in Table \ref{table:estimated-volume-M_+ with Delta>0}.} }
\label{table:estimated-ratio-volume with Delta>0}
\begin{tabular}{|c|c|c|c|c|c|}
    \hline
    $n$ &$10^4$ & $10^5$ & $10^6$ & $10^7$\\
     \hline
   $\frac{v(V_{em+})}{v(V_+)}$&{\color{black}0.130435}&{\color{black}0.177083}&{\color{black}0.17959}&{\color{black}0.170429}\\
   \hline
    $\frac{v(V_{\rm{DLC}}\cap V_{em+})}{v(V_{\rm{DLC}}\cap V_+)}$&{\color{black}0.157895}& {\color{black}0.220779}&{\color{black}0.231668} & {\color{black}0.221289}\\
   \hline
   $\frac{v(V_{\rm{DLC}}\cap V_{em+})}{v(V_+)}$&{\color{black}0.130435}& {\color{black}0.177083}& {\color{black}0.178589}&{\color{black}0.17004}\\
   \hline
   $\frac{v(V_{\rm{DLC}}\cap V_{em+})}{v(V_{em+})}$&{\color{black}1}&{\color{black}1}&{\color{black}0.994429}&{\color{black}0.997721}\\
   \hline
   $\frac{v(V_{\rm{DD}}\cap V_{em+})}{v(V_{\rm{DD}}\cap V_+)}$&{\color{black}0.333333}&{\color{black}0.483871}&{\color{black}0.400763}&{\color{black}0.349948}\\
   \hline
   $\frac{v(V_{\rm{DD}}\cap V_{em+})}{v(V_+)}$&{\color{black}0.0434783}&{\color{black}0.078125}&{\color{black}0.052563}&{\color{black}0.0490753}\\
   \hline
    $\frac{v(V_{\rm{DD}}\cap V_{em+})}{v(V_{em+})}$&{\color{black}0.333333}&{\color{black}0.441176}&{\color{black}0.292479}&{\color{black}0.287952}\\
   \hline
\end{tabular}
\end{table}

{\color{black}
An alternative approach for approximating the number of embeddable matrices within the model is to use Algorithm 5.8 in \cite{casanellas2020embedding} to test the embeddability of the sample points. Tables  \ref{table:python sampleset2} and \ref{table:python estimated-ratio-volumes} below are  analogous to Tables \ref{table:estimated-volume-M_+ with Delta>0} and \ref{table:estimated-ratio-volume with Delta>0}, but Table~\ref{table:python sampleset2} was obtained using the sampling method in \cite[Appendix A]{Roca}, while using either Algorithm 5.8 in \cite{casanellas2020embedding} or the inequalities in Theorem \ref{theorem:criteria SSM-embeddability} yields identical results which are provided in Table~\ref{table:python sampleset2} and Table~\ref{table:python estimated-ratio-volumes}.

We used the python implementation of Algorithm 5.8 in \cite{casanellas2020embedding} provided in \cite[Appendix A]{Roca} and modified it to sample on the set of $4 \times 4$ CS Markov matrices with positive eigenvalues.
The original sampling method used in \cite[Appendix A]{Roca} consisted of sampling uniformly on the set of $4\times4$ centrosymmetric-Markov matrices and what we did is keep sampling until we got $n$ samples  (or as many samples as we require) with positive eigenvalues.

Despite the fact that Theorem \ref{theorem:criteria SSM-embeddability} and Algorithm 5.8 in \cite{casanellas2020embedding} were originally implemented using different programming languages (Wolfram Mathematica and Python respectively) and were tested with different sample sets, the results obtained are quite similar as illustrated by Tables \ref{table:estimated-ratio-volume with Delta>0} and \ref{table:python estimated-ratio-volumes}. In fact, when we are applying both Algorithm 5.8 in \cite{casanellas2020embedding} and Theorem \ref{theorem:criteria SSM-embeddability} on the same sample set in Table \ref{table:python sampleset}, we are obtaining identical results which are displayed in Tables \ref{table:python sampleset2} and \ref{table:python estimated-ratio-volumes}.

}

\begin{table}[H]
\caption{\color{black} Number of samples in $V_+, V_{\rm{DLC}}\cap V_+,\mbox{ and }V_{\rm{DD}}\cap V_+$ obtained by using the sampling method in \cite[Appendix A]{Roca}.
}
\centering
\label{table:python sampleset}
{\color{black}
\begin{tabular}{|c|c|c|c|c|}

    \hline
    Samples in $V_+$ &$10^4$ & $10^5$ & $10^6$ & $10^7$\\
   
    \hline
    Samples in $V_{\rm{DLC}}\cap V_+$& 8531 & 85446 & 854709& 8549100\\\hline
    Samples in $V_{\rm{DD}}\cap V_+$&1464& 14538 &144546 &1448720\\\hline
   
\end{tabular}
}
\end{table}

\begin{table}[H]
\centering
\caption{\color{black} Number of samples in $V_{em+}, \ V_{\rm{DLC}}\cap V_{em+}\mbox{ and }V_{\rm{DD}}\cap V_{em+}$ obtained by applying either Theorem \ref{theorem:criteria SSM-embeddability} or the results in \cite{casanellas2020embedding} on the sample set in Table \ref{table:python sampleset}.
}
\label{table:python sampleset2}
{\color{black}
\begin{tabular}{|c|c|c|c|c|}

\hline
        Samples in $V_{em+}$ & 1877 & 18663 &185357&1862413\\
    \hline
    
    Samples in $V_{\rm{DLC}}\cap V_{em+}$& 1869 & 18586& 184555& 1854592\\\hline
    Samples in $V_{\rm{DD}}\cap V_{em+}$& 516 & 5164 &50058& 504304\\
   \hline
   
\end{tabular}}
\end{table}

\begin{table}[H]
\centering
\caption{{\color{black}Relative volumes ratio between the relevant subsets obtained using hit-and-miss method and either Algorithm 5.8 in \cite{casanellas2020embedding} or Theorem \ref{theorem:criteria SSM-embeddability}. The volumes were estimated as the quotient of the sample sizes in Tables \ref{table:python sampleset} and \ref{table:python sampleset2}.}}
\label{table:python estimated-ratio-volumes}
{\color{black}
\begin{tabular}{|c|c|c|c|c|c|}
    \hline
    $n$ &$10^4$ & $10^5$ & $10^6$ & $10^7$\\
     \hline
   $\frac{v(V_{em+})}{v(V_+)}$&$0.1877$&$0.18663$&$0.185357$&$0.1862413$\\
   \hline
    $\frac{v(V_{\rm{DLC}}\cap V_{em+})}{v(V_{\rm{DLC}}\cap V_+)}$& $0.2191$&$0.2175$&$0.2159$&$0.2169$\\
   \hline
   $\frac{v(V_{\rm{DLC}}\cap V_{em+})}{v(V_+)}$& $0.1869$&$0.18586$&$0.184555$&$0.1854592$\\
   \hline
   $\frac{v(V_{\rm{DLC}}\cap V_{em+})}{v(V_{em+})}$&$0.9957$&$0.9959$&$0.99567$& $0.99580$ \\
   \hline
   $\frac{v(V_{\rm{DD}}\cap V_{em+})}{v(V_{\rm{DD}}\cap V_+)}$&$0.3524$&$0.3552$&$0.3463$&$0.3481$\\
   \hline
   $\frac{v(V_{\rm{DD}}\cap V_{em+})}{v(V_+)}$&$0.0516$&$0.05164$&$0.050058$&$0.0504$\\
   \hline
    $\frac{v(V_{\rm{DD}}\cap V_{em+})}{v(V_{em+})}$&$0.2749$&$0.2767$&$0.2701$&$0.2708$\\
   \hline
\end{tabular}}
\end{table}
\renewcommand{\arraystretch}{1}

{\color{black} 
 }

{\color{black} 
It is worth noting that the embeddability criteria given in Theorem \ref{theorem:criteria SSM-embeddability} use inequalities depending on the entries of the matrix, whereas Algorithm 5.8 in \cite{casanellas2020embedding} relies on the computation of its principal logarithm and its eigenvalues and eigenvector, which may cause numerical issues when working with matrices with determinant close to $0$. What is more the computation of logarithms can be computationally expensive. As a consequence, the algorithm implementing the criterion for embeddability arising from Theorem \ref{theorem:criteria SSM-embeddability} is faster. Table \ref{table:python_times} shows the running times for the implementation of both embeddability criteria used to obtain Table \ref{table:python estimated-ratio-volumes}.

\begin{table}[H]
\centering
\caption{\color{black} Running times for the Python implementation of the embeddability criterion arising from Theorem \ref{theorem:criteria SSM-embeddability} and from Algorithm 5.8 in \cite{casanellas2020embedding}. The simulations were run using a computer with 8GB of memory.
}
\label{table:python_times}

{\color{black}
\begin{tabular}{|c|c|c|c|c|}
    \hline
    &$10^4$ & $10^5$ & $10^6$ & $10^7$\\
   \hline 
    Sampling time & 12.5s &121.5s (2 min) & 1222s (20min) & 12141.8s (3h 22min)\\
    Embedding criteria (Theorem \ref{theorem:criteria SSM-embeddability}) &  28.3s & 273.2s (4min 30s)&2703s (45min) & 27413s (7h 37min)\\
    Embedding criteria (Algorithm 5.8) &84.2s &840.5s (15 min) & 8358 (2h 19min) &83786s (23h 16min)\\
   \hline
\end{tabular}}
\end{table}

The Python implementation of Algorithm 5.8 in \cite{casanellas2020embedding} provided in \cite[Appendix A]{Roca}  can also be used to test the embeddability of any $4\times 4$ CS Markov matrix (including those with non-real eigenvalues) without modifying the embeddability criteria. All it takes is a suitable sample set. As hinted in Remark \ref{rk:k_bound}, this would also be possible using the embedability criterion in Theorem~\ref{theorem:criteria SSM-embeddability} together with the boundaries for $k$ provided in \cite[Theorem 5.5]{casanellas2020embedding}. Table \ref{tab:SSvolume} shows the results obtained when applying Algorithm 5.8 in \cite{casanellas2020embedding} to a set of $10^7$ $4\times4$ CS Markov matrices sampled uniformly.}

\begin{table}[H]
  \centering
    \begin{tabular}{|c|c c c |}
    \hline
   & Samples & Embeddable samples & Proportion of embeddable\\  
    \hline
$V_4^{Markov}$ & $10^7$ &  $173455$ &  $0.0173455$  \\
  $V_{\rm{DLC}}$  & $1021195$ & $172380$ &  $0.1688022$\\
  $V_{\rm{DD}}$ & $156637  $ & $49471$ & $0.3158321$    \\
 \hline
  \end{tabular}
  \caption{\label{tab:SSvolume} Embeddable matrices within {\color{black} $4\times 4$ CS Markov matrices and its intersection with DLC matrices and DD matrices.}}
\end{table}

{\color{black}

As most DLC and DD matrices have positive eigenvalues, the proportion of embedabbile matrices within these subsets is almost the same when admitting matrices with non-positive eigenvalues (as in Table \ref{tab:SSvolume} instead of only considering matrices with positive eigenvalues as we did in Tables \ref{table:estimated-ratio-volume with Delta>0} and \ref{table:python estimated-ratio-volumes}. On the other hand, the proportion of $4\times4$ embeddable CS matrices is much smaller in this case.}

\section{Centrosymmetric matrices and generalized Fourier transformation}
\label{sect:centrosymmetric}

In Section~\ref{sect:embed} and \ref{sect:volume} we have seen the embeddability criteria for $4\times 4$ centrosymmetric Markov matrices and the volume of their relevant subsets.
In this section, we are extending this framework to larger matrices. The importance of this extension is relevant to the goal of synthetic biology which aims to expand the genetic alphabet. For several decades, scientists have been cultivating ways to create novel forms of life with basic biochemical components and properties far removed from anything found in nature. In particular, they {\color{black}are} working to expand the number of amino acids which is only possible if they are able to expand the genetic alphabet (see for example \cite{hachimoji}).

\subsection{Properties of centrosymmetric matrices}

For a fixed $n\in\mathbb{N}$, let $V_n$\index{$V_n$ set of $n\times n$ CS matrices} denote the set of all centrosymmetric matrices of order $n.$ Moreover, let $V_n^{Markov}$\index{$V_n^{Markov}$ set of $n\times n$ CS Markov matrices} and $V_n^{rate}$\index{$V_n^{rate}$ set of $n\times n$ CS rate matrices} denote the set of all centrosymmetric Markov and rate matrices of order $n$, respectively.  As a subspace of the set of all $n\times n$ real matrices, for $n$ even, dim$(V_n)=\frac{n^2}{2}$ while for $n$ odd, dim$(V_n)=\lfloor \frac{n}{2}\rfloor(n+1)+1$. We will now mention some geometric properties of the sets $V_n^{Markov}$ and $V_n^{rate}.$  {\color{black}Furthermore, for any real number $x$, $\lfloor x\rfloor$ and  $\lceil x\rceil$ denote the floor and the ceiling function of $x$, respectively.}

\begin{proposition}
 \begin{enumerate}
     \item For $n$ even, $V_n^{Markov}\subseteq \mathbb{R}^{\frac{n(n-1)}{2}}_{\geq 0}$ is a Cartesian product of $\frac{n}{2}$ standard $(n-1)$-simplices and its volume is $\frac{1}{(n-1)!^{\frac{n}{2}}}.$ For $n$ odd,  $V_n^{Markov}\subseteq \mathbb{R}^{\lfloor\frac{n}{2}\rfloor n}_{\geq 0}$ is a Cartesian product of $\lfloor\frac{n}{2}\rfloor$ standard $(n-1)$-simplices and {\color{black} the $\lfloor \frac{n}{2}\rfloor$-simplex with vertices $\{0, \frac{e_i}{2}\}_{1\leq i\leq \lfloor \frac{n}{2}\rfloor}\cup\{e_{\lfloor \frac{n}{2}\rfloor+1}\}$, where $e_i$ is the $i$-th standard unit vector in $\mathbb{R}^n$, which is the vector that has 1 as the $i$-th component and zeros elsewhere}. Hence,  the volume of $V_n^{Markov}$ is {\color{black}$\frac{1}{2^{\lfloor \frac{n}{2}\rfloor}(\lfloor \frac{n}{2}\rfloor)!(n-1)!^{\lfloor\frac{n}{2}\rfloor}}$}.
    \item For $n$ even, $V_n^{rate}=\mathbb{R}_{\geq 0}^{\frac{n(n-1)}{2}}$ and for $n$ odd,  $V_n^{rate}=\mathbb{R}_{\geq 0}^{\lfloor\frac{n}{2}\rfloor n}$. 
 \end{enumerate}
 \end{proposition}
 \begin{proof}
 {\color{black}Here we consider the following identification for an $n\times n$ centrosymmetric matrix $M$. For $n$ even, $M$ can be thought as a point $(M_1,\dots, M_{\frac{n}{2}})\in (\mathbb{R}^n_{\geq 0})^{\frac{n}{2}}$ where the point $M_i\in \mathbb{R}^n_{\geq 0}$ corresponds to the $i$-th row of $M$. Similarly, for $n$ odd, we identify $M$ as a point in $(\mathbb{R}^n_{\geq 0})^{\lfloor\frac{n}{2}\rfloor}\times \mathbb{R}^{\lfloor\frac{n}{2}\rfloor+1}_{\geq 0}$. Since $M$ is a Markov matrix, under this identification, each point $M_i$ lies in some simplices. Therefore, $V^{Markov}_n$ is a Cartesian product of some simplices. For $n$ even, these simplices are the standard $(n-1)$-dimensional simplex:
 \begin{equation}\label{Eq:standard}
    \left\{
\begin{array}{ll}
      x_1+\cdots+x_n=1,\\
      x_i\geq 0,\quad 1\leq i\leq n
\end{array} 
\right.
\quad \Leftrightarrow \quad
\left\{
\begin{array}{ll}
     x_1+\cdots+x_{n-1}\leq 1,\\
     x_i\geq 0,\quad 1\leq i\leq n-1
\end{array}
\right.  
 \end{equation}
 \noindent For $n$ odd and $1\leq i\leq \lfloor \frac{n}{2}\rfloor$, the point $M_i$ belongs to standard $(n-1)$-simplex above and the point $M_{\lfloor\frac{n}{2}\rfloor+1}$ belongs to the simplex 
 \begin{equation}\label{eq:scaled simplex}
 \left\{
\begin{array}{ll}
      2x_1+\cdots+2x_{\lfloor \frac{n}{2}\rfloor}+x_{\lfloor \frac{n}{2}\rfloor+1}=1\\
      x_i\geq 0,\quad 1\leq i\leq \lfloor \frac{n}{2}\rfloor+1
\end{array} 
\right.
\quad \Leftrightarrow \quad
\left\{
\begin{array}{ll}
     x_1+\cdots+x_{\lfloor\frac{n}{2}\rfloor}\leq \frac{1}{2},\\
     x_i\geq 0,\quad 1\leq i\leq \lfloor\frac{n}{2}\rfloor
\end{array}
\right.
\end{equation}
 
 We now compute the volume of $V^{Markov}_n$. Let us recall the fact that the volume of the Cartesian product of spaces is equal to the product of volumes of each factor space if the volume of each factor space is bounded. Moreover, the $(n-1)$-dimensional volume of the standard simplex in Equation~(\ref{Eq:standard}) in $\mathbb{R}^{n-1}$ is $\frac{1}{(n-1)!}.$ For $n$ even, the statement follows immediately. For $n$ odd, we use the fact that the  $\lfloor\frac{n}{2}\rfloor$-dimensional volume of the simplex in Equation~(\ref{eq:scaled simplex}) is $\frac{1}{2^{\lfloor\frac{n}{2}\rfloor}(\lfloor\frac{n}{2}\rfloor)!}.$
 
 For the second statement, we use the fact that if $Q$ is a rate matrix, then $q_{ii}=-\sum_{j\neq i}q_{ij}$ where $q_{ij}\geq 0$ for $i\neq j$. }
 \end{proof}

In the rest of this section, let $J_n$ be the $n\times n$ anti-diagonal matrix, i.e. the $(i,j)$-entries are one if  $i+j= n+1$ and zero otherwise. The following proposition provides some properties of the matrix $J_n$ that can be checked easily.

\begin{proposition}
   Let $A=(a_{ij})\in M_n(\mathbb{R})$. Then
  \begin{enumerate}
      \item  
 $(AJ_n)_{ij}=a_{i,n+1-j}\mbox{ and }(J_nA)_{ij}=a_{n+1-i,j}.$
 \item $A$ is a centrosymmetric matrix if only if $J_nAJ_n=A.$
  \end{enumerate}
 \end{proposition}

In Section~\ref{sect:embed}, we have seen that $4\times 4$ CS matrices can be block-diagonalized through the matrix $S$. Now we will present a construction of generalized Fourier matrices to block-diagonalize any centrosymmetric matrices. Let us consider the following recursive construction of the $n\times n$ matrix $S_n$:
\begin{align}
\label{Sn}
    S_1=\begin{pmatrix}
    1\\
    \end{pmatrix},
    S_2=\begin{pmatrix}
    1&1\\
    1&-1\\
    \end{pmatrix}\mbox{ and }
    S_n:=\begin{pmatrix}
    1&0&1\\
    0&S_{n-2}&0\\
    1&0&-1\\
    \end{pmatrix}, \mbox{ for }n\geq 3.
\end{align}
 \begin{proposition}
 For each natural number $n\geq 3$, $S_n$ is invertible and its inverse is given by
 $$S_n^{-1}=\begin{pmatrix}
 \frac{1}{2}&0& \frac{1}{2}\\
 0&S_{n-2}^{-1}&0\\
  \frac{1}{2}&0& -\frac{1}{2}\\
 \end{pmatrix}.$$
 \end{proposition}
 \begin{proof}
The proposition easily follows from the definition of $S_n$. Namely,
 $$\begin{pmatrix}
 \frac{1}{2}&0& \frac{1}{2}\\
 0&S_{n-2}^{-1}&0\\
  \frac{1}{2}&0& -\frac{1}{2}\\
 \end{pmatrix}S_n=\begin{pmatrix}
 \frac{1}{2}&0& \frac{1}{2}\\
 0&S_{n-2}^{-1}&0\\
  \frac{1}{2}&0& -\frac{1}{2}\\
 \end{pmatrix}\begin{pmatrix}
    1&0&1\\
    0&S_{n-2}&0\\
    1&0&-1\\
    \end{pmatrix}=I_n.$$ 
    {\color{black}}
 \end{proof}
 
  The following proposition provides another block decomposition of the matrix $S_n$ and its inverse.
 \begin{proposition}\label{prop:partition S_n}
  Let $n\geq 2.$
 \begin{enumerate}
     \item For $n$ even, $S_n=\begin{pmatrix}
     I_{\frac{n}{2}}&J_{\frac{n}{2}}\\
     J_{\frac{n}{2}}&-I_{\frac{n}{2}}\\
     \end{pmatrix}$,
     while for $n$ odd, $S_n=\begin{pmatrix}
     I_{\lfloor\frac{n}{2}\rfloor}&0&J_{\lfloor\frac{n}{2}\rfloor}\\
     0&1&0\\
     J_{\lfloor\frac{n}{2}\rfloor}&0&-I_{\lfloor\frac{n}{2}\rfloor}\\
     \end{pmatrix}.$
     \item Using these block partitions, $S_n^{-1}=\frac{1}{2}S_n$ for $n$ even, while 
      $S_n^{-1}=\begin{pmatrix}
     \frac{1}{2}I_{\lfloor\frac{n}{2}\rfloor}&0& \frac{1}{2}J_{\lfloor\frac{n}{2}\rfloor}\\
     0&1&0\\
      \frac{1}{2}J_{\lfloor\frac{n}{2}\rfloor}&0&- \frac{1}{2}I_{\lfloor\frac{n}{2}\rfloor}\\
     \end{pmatrix}$  for $n$ odd.
 \end{enumerate}
 \end{proposition}
 \begin{proof}
 The proof follows from induction on $n$ and the fact that $J_n^2=I_n$.
 \end{proof}
 
 We will call a vector $v\in \mathbb{R}^n$ \textit{symmetric}
 if $v_i=v_{n+1-i}$ for every $1\leq i\leq n$, i.e. $J_n v=v.$ Moreover, we call a vector $w\in \mathbb{R}^n$ \textit{anti-symmetric}
 if $v_i=-v_{n+1-i}$ for every $1\leq i\leq n$, i.e. $J_n v=-v.$ The following technical proposition will be used in what follows in order to simplify a centrosymmetric matrix.
 \begin{proposition}\label{prop:property of symmetric and anti-symmetric}
 Let $n\geq 2$. Let $v\in \mathbb{R}^n$ be a symmetric vector and  $w\in \mathbb{R}^n$ be an anti-symmetric vector.
 \begin{enumerate}
     \item The last $\lfloor\frac{n}{2}\rfloor$ entries of $S_nv$ and $v^TS_n$ are zero. Similarly,  the last $\lfloor\frac{n}{2}\rfloor$ entries of $S_n^{-1}v$ and $v^TS_n^{-1}$ are zero.
     \item The first $\lfloor\frac{n}{2}\rfloor$ entries of $S_nw$ and $w^TS_n$ are zero. Similarly,  the first $\lfloor\frac{n}{2}\rfloor$ entries of $S_n^{-1}w$ and $w^TS_n^{-1}$ are zero.
     \item Then the sum of the entries of $S_nv$ {\color{black}and} $v^TS_n$ is the sum of the entries of $v$.
     \item Then the sum of the entries of $S_n^{-1}v$ {\color{black}and} $v^TS_n^{-1}$ is the sum of the first $\lceil\frac{n}{2}\rceil$ entries of $v.$
 \end{enumerate}
 \end{proposition}
 \begin{proof}
 We will only prove the first part of item (1) in the proposition using mathematical induction on $n.$ The base case for $n=2$ can be easily obtained. Suppose {\color{black}now that} the proposition holds for all $k<n.$ Let $v=\begin{pmatrix}v_1\\v'\\v_1
 \end{pmatrix}\in \mathbb{R}^n$ be a symmetric element. Then {\color{black}$v'\in \mathbb{R}^{n-2}$} is also symmetric. By direct computation we obtain
 $$S_nv=\begin{pmatrix}
    1&0&1\\
    0&S_{n-2}&0\\
    1&0&-1\\
    \end{pmatrix}\begin{pmatrix}
    v_1\\
    v'\\
    v_1\\
    \end{pmatrix}=\begin{pmatrix}
    2v_1\\S_{n-2}v'\\0\\
    \end{pmatrix}.$$
    The last $\lfloor \frac{n-2}{2}\rfloor$ entries of $S_{n-2}v'$ are zero. Thus, the last $\lfloor \frac{n-2}{2}\rfloor+1=\lfloor \frac{n}{2}\rfloor$ entries of $S_nv$ are zero as well. The proof of the other statements can be obtained analogously using induction.  {\color{black}In particular, let us note that the proof given for item (1) directly implies item (3).} {\color{black}}
 \end{proof}
  
  For a fixed number $n$, let us define the following  map:
  \begin{eqnarray*}
      F_n:& M_n(\mathbb{R}) &\rightarrow M_n(\mathbb{R})\\
      &A &\mapsto F_n(A):=S_n^{-1}AS_n.
  \end{eqnarray*}
   For $n=4,$ we have seen that if $A$ is a CS matrix, then $F_4(A)$ is a block-diagonal matrix where each block is of size $2\times 2$ and is given by $A_1$ and $A_2$. Moreover, the upper block is a Markov matrix. The following lemma provides a generalization to these results.

 \begin{lemma}\label{lemma:block diagonalization}
 Let $n\geq 2.$ Given an $n\times n$ CS matrix $A,$ $F_n(A)$ is the following block-diagonal matrix
 $$F_n(A)=\textup{diag}(A_1,A_2),$$
 where $A_1$ is a matrix of size $\lceil\frac{n}{2}\rceil\times \lceil\frac{n}{2}\rceil$. Furthermore, if $A$ is a Markov (rate) matrix, then $A_1$ is also a Markov (rate) matrix. 
 \end{lemma}

 \begin{proof}
 
  First suppose that $n$ is even. By \cite[Lemma 2]{cantoni1976eigenvalues}, we can partition $A$ into the following block matrices:
  $$A=\begin{pmatrix}
  B_1&B_2\\
  J_{\frac{n}{2}}B_2J_{\frac{n}{2}}&J_{\frac{n}{2}}B_1J_{\frac{n}{2}}\\
  \end{pmatrix},$$
  where $B_1$ and $B_2$ are of size $\lfloor \frac{n}{2}\rfloor \times \lfloor \frac{n}{2}\rfloor.$ By Proposition~\ref{prop:partition S_n}, we have 
  \begin{align*}
      S_n^{-1}AS_n&=\frac{1}{2}\begin{pmatrix}
     I_{\frac{n}{2}}&J_{\frac{n}{2}}\\
     J_{\frac{n}{2}}&-I_{\frac{n}{2}}\\
     \end{pmatrix}
     \begin{pmatrix}
  B_1&B_2\\
  J_{\frac{n}{2}}B_2J_{ \frac{n}{2}}&J_{\frac{n}{2}}B_1J_{ \frac{n}{2}}\\
  \end{pmatrix}
     \begin{pmatrix}
     I_{\frac{n}{2}}&J_{\frac{n}{2}}\\
     J_{\frac{n}{2}}&-I_{\frac{n}{2}}\\
     \end{pmatrix}\\
     &=\begin{pmatrix}
     B_1+B_2J_{\frac{n}{2}}&0\\
     0&J_{\frac{n}{2}}B_1J_{\frac{n}{2}}-J_{\frac{n}{2}}B_2\\
     \end{pmatrix}.\\
  \end{align*}
 Choose $A_1=B_1+B_2J_{\frac{n}{2}}$. Now suppose that $A$ is a Markov matrix. This means that each row of $A$ sums to 1 and $A$ has non-negative entries. Therefore, for $1\leq k\leq \frac{n}{2}$, we have 
  $$\sum_{j=1}^{\frac{n}{2}}(A_1)_{kj}=\sum_{j=1}^{\frac{n}{2}}(B_1+B_2J_{\frac{n}{2}})_{kj}=\sum_{j=1}^{\frac{n}{2}}(a_{kj}+a_{k,\frac{n}{2}+j})=\sum_{j=1}^{n}a_{kj}=1$$
  and for $1\leq j\leq \frac{n}{2}$,
  $(B_1+B_2J_{\frac{n}{2}})_{kj}=a_{kj}+a_{k,\frac{n}{2}+j}\geq 0.$

  Now we {\color{black}consider} the case when $n$ is odd. Again by \cite[Lemma 2]{cantoni1976eigenvalues}, we can  partition $A$ into the following block matrices:
  $$A=\begin{pmatrix}
  B_1&p&B_2\\
  q&r&q J_{\lfloor\frac{n}{2}\rfloor}\\
  J_{\lfloor\frac{n}{2}\rfloor}B_2J_{\lfloor\frac{n}{2}\rfloor}&J_{\lfloor\frac{n}{2}\rfloor}p&J_{\lfloor\frac{n}{2}\rfloor}B_1J_{\lfloor\frac{n}{2}\rfloor}\\
  \end{pmatrix},$$
  where $B_1, B_2 \in M_{\lfloor \frac{n}{2}\rfloor \times \lfloor \frac{n}{2}\rfloor}(\mathbb{R}) $, $p\mbox{ and }q\in M_{1\times\lfloor \frac{n}{2}\rfloor }(\mathbb{R})$ and $r\in  M_{1 \times 1}(\mathbb{R}).$
 By Proposition~\ref{prop:partition S_n}, we have 
 \begin{align*}
     S_n^{-1}AS_n&=\begin{pmatrix}
     \frac{1}{2}I_{\lfloor\frac{n}{2}\rfloor}&0& \frac{1}{2}J_{\lfloor\frac{n}{2}\rfloor}\\
     0&1&0\\
      \frac{1}{2}J_{\lfloor\frac{n}{2}\rfloor}&0&- \frac{1}{2}I_{\lfloor\frac{n}{2}\rfloor}\\
     \end{pmatrix}\begin{pmatrix}
  B_1&p&B_2\\
  q&r&q J_{\lfloor\frac{n}{2}\rfloor}\\
  J_{\lfloor\frac{n}{2}\rfloor}B_2J_{\lfloor\frac{n}{2}\rfloor}&J_{\lfloor\frac{n}{2}\rfloor}p&J_{\lfloor\frac{n}{2}\rfloor}B_1J_{\lfloor\frac{n}{2}\rfloor}\\
  \end{pmatrix}\begin{pmatrix}
     I_{\lfloor\frac{n}{2}\rfloor}&0&J_{\lfloor\frac{n}{2}\rfloor}\\
     0&1&0\\
     J_{\lfloor\frac{n}{2}\rfloor}&0&-I_{\lfloor\frac{n}{2}\rfloor}\\
     \end{pmatrix}\\
     &=\begin{pmatrix}
     B_1+B_2J_{\lfloor\frac{n}{2}\rfloor}&p&0\\
     2q&r&0\\
     0&0&J_{\lfloor\frac{n}{2}\rfloor}B_1J_{\lfloor\frac{n}{2}\rfloor}-J_{\lfloor\frac{n}{2}\rfloor}B_2\\
     \end{pmatrix}.\\
 \end{align*}
 In this case, choose
  $A_1=\begin{pmatrix}
  B_1+B_2J_{\lfloor\frac{n}{2}\rfloor}&p\\
  2q&r\\
  \end{pmatrix}.$ Suppose that $A$ is a Markov matrix. Since each row of $A$ sums to 1, we have
  $$\sum_{j=1}^{\lfloor\frac{n}{2}\rfloor}2q_{1j}+r=\sum_{j=1}^{n}a_{\lfloor\frac{n}{2}\rfloor+1,j}=1$$
  and for $1\leq k\leq \lfloor\frac{n}{2}\rfloor$,
  $$\sum_{j=1}^{\lfloor\frac{n}{2}\rfloor}(B_1+B_2J_{\lfloor\frac{n}{2}\rfloor})_{kj}+p_{k1}=\sum_{j=1}^{\lfloor\frac{n}{2}\rfloor}(a_{kj}+a_{k,\lfloor\frac{n}{2}\rfloor+j+1})+a_{k,\lfloor\frac{n}{2}\rfloor+1}=\sum_{j=1}^{n}a_{kj}=1.$$
 From the fact that the entries of $A$ are non-negative, for $1\leq k, j\leq \lfloor\frac{n}{2}\rfloor$, we obtain that 
 $$(B_1+B_2J_{\lfloor\frac{n}{2}\rfloor})_{kj}=a_{k,j}+a_{k, \lfloor\frac{n}{2}\rfloor+j}\geq 0.$$ 
  Therefore, all entries of $A_1$ sum to 1 and are non-negative meaning that $A_1$ is a Markov matrix as well. We can proceed similarly for the case when $A$ is a rate matrix.
  {\color{black}}
 \end{proof}

 \begin{lemma}\label{lemma:inverse for even}
For any natural number $n$, {\color{black}let} $A_1=(\alpha_{i,j}),$ $A_2=(\beta_{i,j}) \in M_{\lceil\frac{n}{2}\rceil\times \lceil\frac{n}{2}\rceil}(\mathbb{R})$. Suppose that $Q=diag(A_1,A_2)$ is a block diagonal matrix. Then
 \begin{enumerate}
     \item $F_n^{-1}(Q):=S_nQS_n^{-1}$ is a CS matrix.
     \item $F_n^{-1}(Q)$ is a Markov matrix if and only if $A_1$ is a {\color{black}Markov} matrix and for any $1\leq i,j\leq \lfloor\frac{n}{2}\rfloor,$
     $$\alpha_{ij}+\beta_{\lfloor\frac{n}{2}\rfloor+1-i,\lfloor\frac{n}{2}\rfloor+1-j}\geq 0\mbox{ and }\alpha_{i,\lfloor\frac{n}{2}\rfloor+1-j}-\beta_{\lfloor\frac{n}{2}\rfloor+1-i,j}\geq 0.$$
     
     \item $F_n^{-1}(Q)$ is a rate matrix if and only if $A_1$ is a rate matrix and  for any $1\leq i,j\leq \lfloor\frac{n}{2}\rfloor,$ such that for $i=j$,
     $\alpha_{ii}+\beta_{\lfloor\frac{n}{2}\rfloor+1-i,\lfloor\frac{n}{2}\rfloor+1-i}\leq 0$  and for $i\neq j$,
     $$\alpha_{ij}+\beta_{\lfloor\frac{n}{2}\rfloor+1-i,\lfloor\frac{n}{2}\rfloor+1-j}\geq 0\mbox{ and }\alpha_{i,\lfloor\frac{n}{2}\rfloor+1-j}-\beta_{\lfloor\frac{n}{2}\rfloor+1-i,j}\geq 0.$$
     
 \end{enumerate}
 \end{lemma}
 \begin{proof}
 We will only prove the lemma for $n$ even. Similar arguments will work for $n$ odd as well. By Proposition~\ref{prop:partition S_n},
     \begin{align*}
         F_n^{-1}(Q)=\frac{1}{2}
         \begin{pmatrix}
         I_{\frac{n}{2}}&J_{\frac{n}{2}}\\
         J_{\frac{n}{2}}&-I_{\frac{n}{2}}\\
         \end{pmatrix}
         \begin{pmatrix}
         A_1&0\\0&A_2\\
         \end{pmatrix}
         \begin{pmatrix}
         I_{\frac{n}{2}}&J_{\frac{n}{2}}\\
         J_{\frac{n}{2}}&-I_{\frac{n}{2}}\\
         \end{pmatrix}
         &=\frac{1}{2}
         \begin{pmatrix}
         A_1+J_{\frac{n}{2}}A_2J_{\frac{n}{2}}&A_1J_{\frac{n}{2}}-J_{\frac{n}{2}}A_2\\
         J_{\frac{n}{2}}A_1-A_2J_{\frac{n}{2}}&J_{\frac{n}{2}}A_1J_{\frac{n}{2}}+A_2\\
         \end{pmatrix}.
     \end{align*}
     Since $J_{\frac{n}{2}}(A_1+J_{\frac{n}{2}}A_2J_{\frac{n}{2}})J_{\frac{n}{2}}=J_{\frac{n}{2}}A_1J_{\frac{n}{2}}+A_2$ and $J_{\frac{n}{2}}(A_1J_{\frac{n}{2}}-J_{\frac{n}{2}}A_2)J_{\frac{n}{2}}= J_{\frac{n}{2}}A_1-A_2J_{\frac{n}{2}}$, then by \cite[Lemma 2]{cantoni1976eigenvalues}, $F_n^{-1}(Q)$ is centrosymmetric which proves (1). 
     For $1\leq i\leq \frac{n}{2}$, 
     \begin{align*}
        \sum_{j=1}^n(F_n^{-1}(Q))_{ij}&=\frac{1}{2} \sum_{j=1}^n(\alpha_{i,j}+\beta_{\frac{n}{2}+1-i,\frac{n}{2}+1-j}+\alpha_{i,\frac{n}{2}+1-j}-\beta_{\frac{n}{2}+1-i,j})=\sum_{j=1}^n\alpha_{ij}.
     \end{align*}
{\color{black}The above equality means that for $1\leq i\leq \frac{n}{2}$, the $i$-th row sum of $F_n^{-1}(Q)$ and $A_1$ coincide. This implies that if $F_n^{-1}(Q)$ is a Markov (rate) matrix, then $A_1$ is a Markov (rate) matrix as well. Additionally, note that $$( A_1+J_{\frac{n}{2}}A_2J_{\frac{n}{2}})_{ij}=\alpha_{i,j}+\beta_{\frac{n}{2}+1-i,\frac{n}{2}+1-j}\mbox{ and }(A_1J_{\frac{n}{2}}-J_{\frac{n}{2}}A_2)_{ij} =\alpha_{i,\frac{n}{2}+1-j}-\beta_{\frac{n}{2}+1-i,j}.$$ }  Hence, (2) and (3) will follow immediately.{\color{black}}
 \end{proof}

 \subsection{Logarithms of centrosymmetric matrices}

For the special structure encoded by the centrosymmetric matrices, one may ask whether they have logarithms which are also centrosymmetric. In this section, we provide some answers to this question. 

\begin{theorem}\label{thm:necessary and sufficient conditions for centrosymmetric embeddability}
 Let $A\in M_n(\mathbb{R})$ be a CS matrix. Then $A$ has a CS logarithm if and only if both {\color{black}the upper block matrix $A_1$ and the lower block matrix $A_2$} in Lemma~\ref{lemma:block diagonalization} admit a logarithm.
 \end{theorem}
 \begin{proof}
 Suppose that $A$ has a centrosymmetric logarithm $Q$. By Lemma~\ref{lemma:block diagonalization}, 
    $F_n(A)=\mbox{diag}(A_1,A_2)$ and $F_n(Q)=\mbox{diag}(Q_1,Q_2).$
  Then $\exp(Q)=A$ implies that $\exp(Q_1)=A_1$ and $\exp(Q_2)=A_2.$ Hence, $A_1$ and $A_2$ admit a logarithm. Conversely, suppose that $A_1$ and $A_2$ admit a logarithm $Q_1$ and $Q_2$, respectively. Then the matrix $\mbox{diag}(Q_1,Q_2)$ is a logarithm of the matrix $\mbox{diag}(A_1,A_2)$. By Lemma~\ref{lemma:inverse for even}, the matrix $F_n^{-1}(\mbox{diag}(Q_1,Q_2))$ is a centrosymmetric logarithm of $A$. {\color{black}}
 \end{proof}

 \begin{proposition}
 Let $A\in M_n(\mathbb{R})$ be a CS matrix. If $A$ is invertible, then it has infinitely many CS logarithms.
 \end{proposition}
 \begin{proof}
 The assumptions imply that the matrices $A_1$ and $A_2$ in Lemma~\ref{lemma:block diagonalization} are invertible. By \cite[Theorem 1.28]{Higham}, each $A_1$ and $A_2$ has infinitely many logarithms. Hence, Theorem~\ref{thm:necessary and sufficient conditions for centrosymmetric embeddability} {\color{black}implies}  that $A$ has infinitely many centrosymmetric logarithms. {\color{black}}
 \end{proof}

 \begin{proposition}
 Let $A\in M_n(\mathbb{R})$ be a CS matrix such that $Log(A)$ is well-defined. Then $Log(A)$ is again centrosymmetric.
 \end{proposition}
\begin{proof}
Let us suppose that $Log(A)$ is not centrosymmetric matrix. Define the matrix $Q=J_n(Log(A))J_n$. Then $Q\neq Log(A)$ since $Log(A)$ is not centrosymmetric. It is also clear that $\exp(Q)=A$. Moreover, since $J_n^2=I_n$, the matrices $Log(A)$ and $Q$ have the same eigenvalues. Therefore, $Q$ is also a principal logarithm of $A$, a contradiction to the uniqueness of principal logarithm. Hence, $Log(A)$ must be centrosymmetric. {\color{black}}
\end{proof}

The following theorem characterizes the logarithms of any invertible CS Markov matrices.
\begin{theorem}\label{thm:LogDecomp}
 Let $A\in M_n(\mathbb{R})$ be an invertible CS Markov matrix.  {\color{black}Let $A_1=N_1D_1N_1^{-1}$ where $D_1=diag(R_1,R_2,\dots, R_l)$ is a Jordan form of $A_1$, where $A_1$ is the upper block matrix in Lemma~\ref{lemma:block diagonalization}. Similarly, let $A_2=N_2D_2N_2^{-1}$ where $D_2=diag(T_1,T_2,\dots, T_l)$ is a Jordan form of $A_2$ where $A_1$ is the lower block matrix in Lemma~\ref{lemma:block diagonalization}}. Then $A$ has a countable infinitely many logarithms given by
 $$Q:=S_n N D N^{-1}S_n^{-1},$$
 where 
 $$N:=\textup{diag}(N_1,N_2)\quad\mbox{ and }\quad D:=\textup{diag}(D_1',D_2'),$$
 {\color{black}and} $D_i'$ denotes a logarithm of $D_i$. {\color{black} In particular, these logarithms of $A$ are primary functions of $A$.}
\end{theorem}
\begin{proof}
The theorem follows immediately from \cite[Theorem 1.28]{Higham}. {\color{black}}
\end{proof}
{\color{black}For the definition of primary function of a matrix, we refer the reader to \cite{Higham}. The above theorem says that the logarithms of a nonsingular centrosymmetric matrix contains a countable infinitely many primary logarithms and they are centrosymmetric matrices as well. }

 Finally, we will present a necessary condition for embeddability of CS Markov matrices in higher dimensions.
 \begin{lemma}\label{lemma:necessary condition}
 Let $n\geq 2$. Suppose that $A=(a_{ij})$ is {\color{black}an embeddable CS Markov matrix of size $n\times n$ with a CS logarithm}. Then for $n$ even,
 $$\sum_{j=1}^{\frac{n}{2}}(a_{jj}+a_{j,n-j+1})>1, $$
 while for $n$ odd,
 $$\sum_{j=1}^{\lfloor\frac{n}{2}\rfloor}(a_{jj}+a_{j,n-j+1})+a_{\lfloor\frac{n}{2}\rfloor+1,\lfloor\frac{n}{2}\rfloor+1}>1.$$
 \end{lemma}
 
 \begin{proof}
 Since $A$ is  {\color{black}an embeddable matrix with CS logarithm}, we write $A=\exp(Q)$ for some CS rate matrix $Q$, and then 
 $${\color{black}F_n(A)}=F_n(\exp(Q))=\exp(F_n(Q)).$$
 By Lemma~\ref{lemma:block diagonalization}, for the centrosymmetric matrices $A, Q$, we have  $F_n(A)=\textup{diag}(A_1,A_2)$
 and 
 $F_n(Q)=\textup{diag}(Q_1,Q_2)$
 where $A_1$ is a Markov matrix and {\color{black}$Q_1$} is a rate matrix of size $\lceil\frac{n}{2}\rceil\times \lceil\frac{n}{2}\rceil.$ Therefore, $A_1=\exp(Q_1).$ If $\lambda_1,\cdots, \lambda_{\lceil\frac{n}{2}\rceil}$ are the eigenvalues, perhaps not distinct, of $Q_1,$ then the eigenvalues of $A_1$ are $e^{\lambda_1},\cdots,e^{\lambda_{\lceil\frac{n}{2}\rceil}}$. Since one of $\lambda_i$'s is zero, then the trace of $A_1$ which is the sum of its eigenvalues is equal to
 $$tr(A_1)=\sum_{j=1}^{\lceil\frac{n}{2}\rceil}e^{\lambda_j}>1.$$
 We now need to show that trace of $A_1$ has the form written in the lemma. Suppose that $n$ is even. By the proof of Lemma~\ref{lemma:block diagonalization}, then
 $$tr(A_1)=\sum_{j=1}^{\frac{n}{2}}(B_1+B_2J_{\frac{n}{2}})_{jj}=\sum_{j=1}^{\frac{n}{2}}(a_{jj}+a_{j,\frac{n}{2}+j})=\sum_{j=1}^{\frac{n}{2}}(a_{jj}+a_{j,n-j+1}).$$
The proof for odd $n$ can be obtained similarly.{\color{black}}
 \end{proof}
 
 Let $X_n\subseteq V_n^{Markov}$ be the subset containing all centrosymmetric-embeddable Markov matrices. We want to obtain an upper bound of the volume of $X_n$ using Lemma~\ref{lemma:necessary condition}. Let $Y_n\subseteq V_n^{Markov}$ be the subset containing all centrosymmetric Markov matrices such that after applying the generalized Fourier transformation, the trace of the upper block matrix is greater than 1. The previous lemma {\color{black}implies} that $X_n\subseteq Y_n$ and hence, $v(X_n)\leq v(Y_n)$. Moreover, the upper bound $v(Y_n)$ is easy to compute as $Y_n$ is a polytope and for some values of $n$, these volumes are presented in Table~\ref{table:exact volume Y_n}. 
 We see from Table~\ref{table:exact volume Y_n}, there are at most $50\%$ of matrices in $V_4^{Markov}$ that are centrosymmetically-embeddable and hence this upper bound $v(Y_4)$ is not good. For $n=5$,  approximately, there are at most $62\%$ in $V_4^{Markov}$ that are centrosymmetrically-embeddable but for $n=6$, this upper bound gives a better proportion, which is approximately $0.1\%$.
  \begin{table}[H]
      \centering
      \begin{tabular}{|c|c|c|c|}
      \hline
         & dimension of $V_n^{Markov}$  &$v(Y_n)$& $v(V_n^{Markov})$\\
         \hline
           $n=4$&6&$\frac{1}{72}\approx 1.39\times 10^{-2}$&$\frac{1}{36}\approx 2.78\times 10^{-2}$\\
           \hline
           $n=5$&10&$\frac{653}{4838400}\approx 1.35\times 10^{-4}$&$\frac{1}{4608}\approx 2.17\times 10^{-4}$\\
           \hline
           $n=6$&15&$\frac{433}{653837184000}\approx 6.22 \times 10^{-10}$&$\frac{1}{1728000}\approx 5.79\times 10^{-7}$\\
           \hline
      \end{tabular}
      \caption{The exact volume $v(Y_n), n\in\{4,5,6\}$ computed using \texttt{Polymake}.}
      \label{table:exact volume Y_n}
  \end{table}

\section{Embeddability of $6\times 6$ centrosymmetric matrices.}
\label{sect:6x6}


Throughout this section we shall consider $A$ to be a $6 \times 6$ centrosymmetric Markov matrix with {\color{black} distinct} eigenvalues. 
In particular, the matrices considered in this section are diagonalizable and are a dense subset of all $6 \times 6$ centrosymetric Markov matrices. Note that this notation differs from the notation for Markov matrices used in previous sections in order to make it consistent with the notation used in the results presented for generic centrosymmetric matrices.\\

In the previous section, we {\color{black}showed}  that $F(A)$ is a block-diagonal real matrix composed of two $3\times 3$ blocks denoted by $A_1$ and $A_2$. Since both $A_1$ and $A_2$ have real entries, each of these matrices has at most one conjugate pair of eigenvalues. Adapting the notation introduced in Theorem~\ref{thm:LogDecomp} to diagonalizable matrices we have  $N_1,N_2 \in GL_3(\mathbb{C})$ such that $A_1=N_1 diag(1,\lambda_1, \lambda_2) N_1^{-1}$ and $A_2=N_2 diag(\mu,\gamma_1, \gamma_2) N_2^{-1}$ with $\mu \in \mathbb{R}_{>0}$ and $\lambda_i,\gamma_i \in \mathbb{C}\setminus \mathbb{R}_{\geq0}$. Moreover, we can assume that $Im(\lambda_1)>0$ without loss of generality (this can be achieved by permuting the second and third columns of $N_1$ if necessary). For ease of reading, we will define as $P:= S_6 diag(N_1,N_2)$, where $S_6$ is the matrix used to obtain the Fourier transform $F(A)$ and was introduced in Section \eqref{Sn}.\\

Next we give a criterion for the embeddability of $A$ for each of the following cases:
\begin{equation}
    \begin{tabular}{c|cc} \label{tab:cases}
      & $\gamma_i \in \mathbb{R}_{>0}$ & $\gamma_i \in \mathbb{C}\setminus \mathbb{R}$\\
      \hline
$\lambda_i \in \mathbb{R}_{>0}$ & case 1 & case 2\\
$\lambda_i \in \mathbb{C}\setminus \mathbb{R}$ & case 3 & case 4\\
    \end{tabular}
\end{equation}

{\color{black}
\begin{proposition}
If a $6\times6$ cetrosymmetric Markov matrix $A$ does not belong to any of the cases in Table \ref{tab:cases}, then it is not embeddable.\end{proposition}
\begin{proof}
   If $A$ satifies the hypothesis of the proposition then either it has a null eigenvalue or it has a simple negative eigenvalue. In the former case $A$ is a singular matrix and hence it has no logarithm.  If $A$ had a simple negative eigenvalue, then all its logarithms would have a non-real eigenvalue whose complementary pair is not an eigenvalue of $A$ (otherwise $M$ would have a repeated eigenvalues). Therefore, $A$ has no real logarithm.
\end{proof}
}

\begin{remark} All the results in this section can be adapted to $5\times 5$ centrosymmetric Markov matrices by not considering the eigenvalue $\mu$ and modifying the forthcoming definitions of the matrices $Log_{-1}(A)$ and $V$ accordingly (i.e. removing the fourth row and column in the corresponding diagonal matrix). In addition, these results still hold if the eigenvalue $1$ of the Markov matrix has multiplicity $2$.
\end{remark}

\subsection*{Case 1}

{\color{black}
The results for this case are not restricted to centrosymmetric matrices but can be applied to decide the embeddability of any suitable Markov.
}

\begin{proposition}
{\color{black}
If all the eigenvalues of a Markov matrix $A$ are distinct and positive, then} 
$A$ is embeddable if and only if $Log(A)$ is a rate matrix.
\end{proposition}
{\color{black}
\begin{proof}
    If $A$ has distinct real eigenvalues then it has only one real logarithm, namely $Log(A)$ (see \cite{Culver}). 
\end{proof}
}
\subsection*{Case 2}

In this case $A$ has exactly one conjugate pair of complex eigenvalues and we obtain the following criterion by adapting Corollary 5.6 in \cite{casanellas2020embedding} to our framework:

\begin{proposition} \label{Prop:Case2}
Given the matrix $V:=  P \;
 diag(0,0,0,0,2\pi i , -2\pi i) \;P^{-1}$ define:
 $$\mathcal{L}:= \displaystyle \max_{(i,j):\ i\neq j,\  V_{i,j} > 0} \left\lceil -\frac{Log(A)_{i,j}}{V_{i,j}}  \right\rceil,  \qquad \mathcal{U}:= \displaystyle \min_{(i,j):\ i\neq j, \  V_{i,j} < 0} \left\lfloor -\frac{Log(A)_{i,j}}{V_{i,j}}
  \right\rfloor
$$

and set  $ \ \mathcal{N}:= \{(i,j): i\neq j, \  V_{i,j}=0 \text{ and  } Log({\color{black}A})_{i,j}<0\}.$ Then, 
\begin{enumerate}
    \item $A$ is embeddable if and only if  $\mathcal{N} = \emptyset$ and $  \mathcal{L}  \leq  \mathcal{U}$.
    \item the set of Markov generators for $A$ is $\left\lbrace 
Q=Log(A) + k V: k\in \mathbb{Z} \text{ such that } \mathcal{L}  \leq k \leq  \mathcal{U} \right\rbrace$.
\end{enumerate}
\end{proposition}
\begin{proof}
 The proof of this theorem is analogous to the proof of Theorem 5.5 in \cite{casanellas2020embeddability} but considering the matrix $V$ as defined here. {\color{black}According to Proposition~\ref{Prop:Logenum}, any Markov generator of $A$ is of the form 
 \begin{align*}
     Log_k(A)&= P diag(0, \log(\lambda_1),\log(\lambda_2), \log(\mu), \log_k(\gamma_1), \overline{\log_k(\gamma_1)} P^{-1}\\
     &= P diag(0, \log(\lambda_1),\log(\lambda_2), \log(\mu), \log_k(\gamma_1)+2\pi ki, \overline{\log_k(\gamma_1)}-2\pi k i) P^{-1}.\\
 \end{align*}
  Such a logarithm can be rewritten as $Log(A)+kV$. Using this, we will prove that $Log_k(A)=Log(A)+kV$ is a rate matrix if and only if $\mathcal{N} = \emptyset$ and $\mathcal{L}  \leq k \leq  \mathcal{U}$.
  
  Suppose that there exists $k\in \mathbb{Z}$ such that $Log_k(A)$ is a rate matrix. Hence, $Log(A)_{i,j}+kV_{i,j}\geq 0$ for all $i\neq j$. For $i\neq j$, we have:
  \begin{itemize}
      \item[(a)] $Log(A)_{i,j}\geq 0$ for all $i\neq j$ such that $V_{i,j}=0$. This means that $\mathcal{N}=\emptyset.$
      \item[(b)]$-\frac{Log(A)_{i,j}}{V_{i,j}}\leq k$ for all $i\neq j$ such that $V_{i,j}>0$. This means that $\mathcal{L}\leq k$.
      \item[(c)]$-\frac{Log(A)_{i,j}}{V_{i,j}}\geq k$ for all $i\neq j$ such that $V_{i,j}<0$. This means that $k\leq \mathcal{U}$.
  \end{itemize} 
  Conversely, suppose that $\mathcal{N} = \emptyset$ and  and that there is $k\in \mathbb{Z}$ such that $\mathcal{L}  \leq k \leq  \mathcal{U}$. We want to check that $Log_k(A)$ is a rate matrix. According to Proposition~\ref{Prop:Logenum}, each row of $Log_k(A)$ sums to $0$. Moreover, for $i\neq j$, we have:
  \begin{itemize}
      \item[(a)] if $V_{i,j}=0$, then $Log_k(A)_{i,j}=Log(A)_{i,j}$. Since $\mathcal{N}=\emptyset$, $Log_k(A)_{i,j}=Log(A)_{i,j}\geq 0$.
      \item[(b)]if $V_{i,j}>0$, then 
      $Log_k(A)_{i,j}=Log(A)_{i,j}+k V_{i,j}\geq Log(A)_{i,j}+\mathcal{L} V_{i,j}\geq Log(A)_{i,j}+(-\frac{Log(A)_{i,j}}{V_{i,j}}) V_{i,j}=0.$
      \item[(c)]if $V_{i,j}<0$, then 
      $-Log_k(A)_{i,j}=-Log(A)_{i,j}-k V_{i,j}\leq -Log(A)_{i,j}-\mathcal{U} V_{i,j}\leq -Log(A)_{i,j}-(-\frac{Log(A)_{i,j}}{V_{i,j}}) V_{i,j}=0.$
  \end{itemize}
  The proof is now complete.
  }
 \end{proof}
\subsection*{Case 3}
As in Case 2, $A$ has exactly one conjugate pair of eigenvalues and hence its embeddability (and all its generators) can be determined by using Proposition~\ref{Prop:Case2} but defining the matrix $V$ as $V=P \;
 diag(0,0,0,0,2\pi i , -2\pi i) \;P^{-1}$. However in Case $3$ the conjugate pair of eigenvalues lie in $A_1$ which is a Markov matrix. This allows us to use the results regarding the embeddability of $3\times 3$ Markov matrices to obtain an alternative criterion  to test the embeddability of $A$. To this end we define \begin{equation}\label{eq:Log-1}
     Log_{-1}(A):= P \; diag(0, z, \overline{z},\log(\mu), \log(\gamma_1)\log(\gamma_2)) \;P^{-1}
 \end{equation} where $z:= \log_{-1}(\lambda_1)$.

 \begin{proposition} \label{prop:Case3}
 {\color{black} The matrix } $A$ is embeddable if and only if $Log(A)$ or $Log_{-1}(A)$ are rate matrices.
  \end{proposition}
 \begin{proof}
 Note that $\exp(Log(A))=\exp(Log_{-1}(A))=A$ so one of the implications is immediate to prove. To prove the other implication, we assume that $A$ is embeddable and let $Q$ be a Markov generator for it. Proposition~\ref{Prop:Logenum} yields that 
 $$Q=P diag(0,\log_{k_1}(\lambda_1),\log_{k_2}(\lambda_2),\log_{k_3}(\mu),\log_{k_4}(\gamma_1),\log_{k_5}(\gamma_2)) \; P^{-1},$$ 
 for some integers $k_1,\dots,k_5 \in \mathbb{Z}$. Therefore, $F(Q)=\begin{pmatrix}
 Q_1 & 0\\
 0 & Q_2\\
 \end{pmatrix}$ where $Q_1$ and $Q_2$ are real {\color{black}logarithms} of $A_1$ and $A_2$ respectively. 

Since $A_2$ is a real matrix with {\color{black} distinct} positive eigenvalues, its only real logarithm is its principal logarithm. This implies that $k_3=k_4=k_5=0$ (so that  $Q_2=Log(A_2)$ ).
 
Now, recall that $A_1$ is a Markov matrix (see Lemma~\ref{lemma:block diagonalization}). Using Proposition~\ref{Prop:Logenum} again, we obtain that $Q_1$ is a rate matrix, thus $A_1$ is embeddable. To conclude the proof it is enough to recall Theorem 4 in \cite{Cuthbert}, which yields that $A_1$ is embeddable if and only if $Log(A_1)$ or $P_1\;diag(0,z,\overline{z}) \;P_1^{-1}$ is a rate matrix. {\color{black}}
 \end{proof}

\subsection*{Case 4}
 
 In this case, the solution to the embedding problem can be obtained as a byproduct of the results for the previous cases:
 
 \begin{proposition} 
 Let $Log_{0,0}(A)$ denote the principal logarithm of $A$ and $Log_{-1,0}(A)$ denote the matrix in \eqref{eq:Log-1}.
Given the matrix $V:=  P \;
 diag(0,0,0,0,2\pi i , -2\pi i) \;P^{-1}$ and $k\in \{0,-1\}$ define:
 $$\mathcal{L}_k:= \displaystyle \max_{(i,j):\ i\neq j,\  V_{i,j} > 0} \left\lceil -\frac{Log_{k,0}(A)_{i,j}}{V_{i,j}}  \right\rceil,  \qquad 
 \mathcal{U}_k:= \displaystyle \min_{(i,j):\ i\neq j, \  V_{i,j} < 0} \left\lfloor -\frac{Log_{k,0}(A)_{i,j}}{V_{i,j}}
  \right\rfloor
$$
and set  $ \ \mathcal{N}_k:= \{(i,j): i\neq j, \  V_{i,j}=0 \text{ and  } Log_{k,0}({\color{black}A})_{i,j}<0\}.$ Then, \begin{enumerate}
    \item 
$A$ is embeddable if and only if  $\mathcal{N}_k = \emptyset$ and $  \mathcal{L}_k  \leq  \mathcal{U}_k$ for $k=0$ or $k=-1$.\\
    
\item If $A$ is embeddable, then at least one of its Markov generator  can be written as {\color{black}$$Log_{k,k_2}(A):=P\;  diag(0, \log_{k}(\lambda_1),\overline{\log_{k}(\lambda_1)}, \log(\mu), \log_{k_2}(\gamma_1), \overline{\log_{k_2}(\gamma_1)} \;P^{-1}$$ with $k\in \{0,-1\}$ and $k_2\in \mathbb{Z}$ such that $\mathcal{L}_{k} \leq k_2 \leq \mathcal{U}_{k}$}.

\end{enumerate}
\end{proposition}
\begin{proof}
{\color{black} The matrix} $A$ is embeddable if and only if it admits a Markov generator. According to Proposition~\ref{Prop:Logenum}, if such a generator {\color{black}$Q$} exists then it can be written as $Log_{k_1,k_2}(A)$ for some $k_1,k_2 \in \mathbb{Z}$.  Therefore, {\color{black}Lemma~\ref{lemma:block diagonalization} implies that 
$F(A)=\begin{pmatrix}
 A_1 & 0\\
 0 & A_2\\
 \end{pmatrix}$ for some matrices $A_1$ and $A_2$}. Moreover,
$F(Q)=\begin{pmatrix}
 Q_1 & 0\\
 0 & Q_2\\
 \end{pmatrix}$ where $Q_1$ and $Q_2$ are real logarithms of $A_1$ and $A_2$ respectively.

 As shown in the proof of Proposition~\ref{prop:Case3}, $A_1$ is actually a Markov matrix and $Q_1$ is a Markov generator for it (see also Lemma~\ref{lemma:block diagonalization}). Moreover, by Theorem 4 in \cite{Cuthbert}, $A_1$ is embeddable if and only if $Log(A_1)$ or $Log_{-1}(A_1)$ are rate matrices. This implies that $Log_{k_1,k_2}(A)$ is a rate matrix if and only if $Log_{0,k_2}(A)$ or $Log_{-1,k_2}$ {\color{black}are rate matrices}. To conclude the proof we proceed as in the proof of Proposition~\ref{Prop:Case2}. Indeed, note that {\color{black}for $k\in \{0,-1\}$,} $Log_{k,k_2}(A) = Log_{k,0}(A) + k_2 V$. Using this, {\color{black}it is immediate to check that $Log_{k,k_2}(A)$ is a rate matrix if and only if $\mathcal{N}_{k} = \emptyset$ and $  \mathcal{L}_{k}  \leq k_2\leq  \mathcal{U}_{k}$}. {\color{black}}
\end{proof}

\section{Discussion}\label{discussion}

 The central symmetry is motivated by the complementarity between both strands of the DNA. When a nucleotide substitution occurs in one strand, there is also a substitution between the corresponding complementary nucleotides on the other strand. Therefore, working with centrosymmetric Markov matrices is the most general approach when considering both DNA strands.
 
 In this paper, we have discussed the embedding problem for centrosymmetric Markov matrices. In Theorem~\ref{theorem:criteria SSM-embeddability}, we have obtained a characterization of the embeddabilty of $4\times 4$ centrosymmetric Markov matrices which are exactly the {\color{black}s}trand {\color{black}s}ymmetric Markov matrices. In particular, we have also shown that if a  $4\times 4$ CS Markov matrix is embeddable, then any of its Markov generators is also a CS matrix. Furthermore, In Section~\ref{sect:6x6}, we have discussed the embeddability criteria for larger centrosymmetric matrices. 
 
 As a consequence of the characterization of Theorem~\ref{theorem:criteria SSM-embeddability}, we have been able to compute and compare the volume of the embeddable $4\times 4$ CS Markov matrices within some subspaces of $4\times 4$ CS Markov matrices. These volume comparisons can be seen in Table~\ref{table:estimated-ratio-volume with Delta>0} and Table~\ref{tab:SSvolume}. For larger matrices, using the results in Section~\ref{sect:6x6}, we have estimated the  proportion of embeddable matrices within the set of all $6\times6$ centrosymmetric Markov matrices and within the subsets of DLC and DD matrices. This is summarized in Table~\ref{Tab:6x6Embed} below. The computations were repeated several times obtaining results with small differences in the values but the same order of magnitude and starting digits.

\begin{table}[H]
  \centering
    \begin{tabular}{|c|c c c |}
    \hline
    Set & Sample points & Embeddable sample points &   Rel. vol. of embeddable matrices \\
   \hline 
   $V^{Markov}_6$& $10^8$& $1370$ &$0.0000137$\\
   $V_{\rm{DLC}}$& $1034607$ & $1362$ & $0.0013164$\\
   $V_{\rm{DD}}$&$3048$&$84$&$0.0275590$\\
    \hline
 
  \end{tabular}
  \caption{\label{Tab:6x6Embed} Relative volume of embeddable matrices within relevant subsets of $6\times6$ centrosymmetric Markov matrices. The results were obtained {\color{black}using}  the hit-and-miss Monte Carlo integration with $10^7$ sample points.}
\end{table}

As we have seen in Section~\ref{sect:embed} and \ref{sect:6x6}, we have only considered in detail the embeddability of CS Markov matrices of size $n=4$ and $n=6$. We expect that the proportion of the embeddable CS Markov matrices within the subset of Markov matrices in larger dimension tends to zero as $n$ grows larger as indicated by Table~\ref{table:estimated-ratio-volume with Delta>0}, \ref{tab:SSvolume}, \ref{table:exact volume Y_n}, and \ref{Tab:6x6Embed}.

These results together with the results obtained for the strand symmetric model (see Table~\ref{tab:SSvolume}) indicate that restricting to homogeneous Markov processes in continuous-time is a very strong restriction because non-embeddable matrices are discarded and their proportion is much larger than that of embeddable matrices. For instance, in the $2\times 2$ case exactly $50\%$ of the matrices are discarded \cite[Table 5]{ardiyansyah2021model}, while in the case of $4\times4$ matrices up to $98.26545\%$ of the matrices are discarded (see Table~\ref{tab:SSvolume}) and in the case of $6\times6$ matrices the amount of discarded matrices is about $99.99863\%$ as indicated in Table~\ref{Tab:6x6Embed}. However, when restricting to  subsets of Markov matrices which are mathematically more meaningful in biological terms, such as DD or DLC matrices, the proportion of embeddable matrices is much higher so that we are discarding less matrices (e.g. for DD we discard $68.41679\%$ of $4\times4$ matrices and $97.2441\%$ of $6\times6$ matrices). This is not to say that it makes no sense to use continuous-time models but to highlight that one should take the above restrictions into consideration when working with these models. Conversely, when working with the whole set of Markov matrices one has to be aware that they might end up considering lots of non-meaningful matrices.
 
\subsection*{Acknowledgements.}
Dimitra Kosta was partially supported by a Royal Society Dorothy Hodgkin Research Fellowship DHF$\backslash$R1$\backslash$201246. Jordi Roca-Lacostena was partially funded by Secretaria d’Universitats i Recerca de la Generalitat de
Catalunya (AGAUR 2018FI\_B\_00947). Muhammad Ardiyansyah is {\color{black}
partially supported by the Academy of Finland Grant No. 323416}.


\bibliographystyle{plain}
\bibliography{references.bib}

\begin{thebibliography}{10}

\bibitem{aitken2017determinants}
Alexander~Craig Aitken.
\newblock {\em Determinants and matrices}.
\newblock Read Books Ltd, 2017.

\bibitem{ardiyansyah2021model}
Muhammad Ardiyansyah, Dimitra Kosta, and Kaie Kubjas.
\newblock The model-specific {M}arkov embedding problem for symmetric
  group-based models.
\newblock {\em Journal of Mathematical Biology}, 83(3):1--26, 2021.

\bibitem{BaakeSumner}
Michael Baake and Jeremy Sumner.
\newblock Notes on markov embedding.
\newblock {\em Linear Algebra and its Applications}, 594:262--299, 2020.

\bibitem{bain1992ribosome}
JD~Bain, Christopher Switzer, Richard Chamberlin, and Steven~A Benner.
\newblock Ribosome-mediated incorporation of a non-standard amino acid into a
  peptide through expansion of the genetic code.
\newblock {\em Nature}, 356(6369):537--539, 1992.

\bibitem{benner2005synthetic}
Steven~A Benner and A~Michael Sismour.
\newblock Synthetic biology.
\newblock {\em Nature Reviews Genetics}, 6(7):533--543, 2005.

\bibitem{cantoni1976eigenvalues}
A~Cantoni and P~Butler.
\newblock Eigenvalues and eigenvectors of symmetric centrosymmetric matrices.
\newblock {\em Linear Algebra and its Applications}, 13(3):275--288, 1976.

\bibitem{Carette1995CharacterizationsOE}
Philippe Carette.
\newblock Characterizations of embeddable 3$\times$3 stochastic matrices with a
  negative eigenvalue.
\newblock {\em New York Journal of Mathematics}, 1:120--129, 1995.

\bibitem{casanellas2020embeddability}
Marta Casanellas, Jes{\'u}s Fern{\'a}ndez-S{\'a}nchez, and Jordi
  Roca-Lacostena.
\newblock Embeddability and rate identifiability of {K}imura 2-parameter
  matrices.
\newblock {\em Journal of Mathematical Biology}, 80(4):995--1019, 2020.

\bibitem{casanellas2020open}
Marta Casanellas, Jes{\'u}s Fern{\'a}ndez-S{\'a}nchez, and Jordi
  Roca-Lacostena.
\newblock An open set of 4$\times$ 4 embeddable matrices whose principal
  logarithm is not a markov generator.
\newblock {\em Linear and Multilinear Algebra}, pages 1--12, 2020.

\bibitem{casanellas2020embedding}
Marta Casanellas, Jesús Fernández-Sánchez, and Jordi Roca-Lacostena.
\newblock The embedding problem for {M}arkov matrices.
\newblock {\em Publicacions Matem\`atiques}, 2022.

\bibitem{casanellas2013generating}
Marta Casanellas and Anna~M Kedzierska.
\newblock Generating {M}arkov evolutionary matrices for a given branch length.
\newblock {\em Linear Algebra and its Applications}, 438(5):2484--2499, 2013.

\bibitem{CasanellasSullivant}
Marta Casanellas and Seth Sullivant.
\newblock The strand symmetric model.
\newblock In Lior Pachter and Bernd Sturmfels, editors, {\em Algebraic
  statistics for computational biology}. Cambridge University Press, New York,
  2005.

\bibitem{chang1996full}
Joseph~T Chang.
\newblock Full reconstruction of {M}arkov models on evolutionary trees:
  identifiability and consistency.
\newblock {\em Mathematical biosciences}, 137(1):51--73, 1996.

\bibitem{ChenChen}
Yong Chen and Jianmin Chen.
\newblock On the imbedding problem for three-state time homogeneous {M}arkov
  chains with coinciding negative eigenvalues.
\newblock {\em Journal of Theoretical Probability}, 24:928–938, 2011.

\bibitem{Culver}
Walter~J. Culver.
\newblock On the existence and uniqueness of the real logarithm of a matrix.
\newblock {\em {Proceedings of the American Mathematical Society}},
  17:1146--1151, 1966.

\bibitem{cuthbert1972uniqueness}
James~R Cuthbert.
\newblock On uniqueness of the logarithm for {M}arkov semi-groups.
\newblock {\em Journal of the London Mathematical Society}, 2(4):623--630,
  1972.

\bibitem{davies2010embeddable}
E~Davies et~al.
\newblock Embeddable {M}arkov matrices.
\newblock {\em Electronic Journal of Probability}, 15:1474--1486, 2010.

\bibitem{Elfving}
Gustav Elfving.
\newblock Zur theorie der {M}arkoffschen ketten.
\newblock {\em Acta Societatis Scientiarum Fennic\ae Nova Series A}, 2(8):17
  pages, 1937.

\bibitem{Fuglede}
B.~Fuglede.
\newblock On the imbedding problem for stochastic and doubly stochastic
  matrices.
\newblock {\em Probability Theory and Related Fields}, 80:241--260, 1988.

\bibitem{gawrilow2000polymake}
Ewgenij Gawrilow and Michael Joswig.
\newblock Polymake: a framework for analyzing convex polytopes.
\newblock In {\em Polytopes—combinatorics and computation}, pages 43--73.
  Springer, 2000.

\bibitem{Goodman}
G.~S. Goodman.
\newblock An intrinsic time for non-stationary finite {M}arkov chains.
\newblock {\em Zeitschrift für Wahrscheinlichkeitstheorie und Verwandte
  Gebiete}, 16:165--180, 1970.

\bibitem{hammersley2013monte}
John Hammersley.
\newblock {\em Monte carlo methods}.
\newblock Springer Science \& Business Media, 2013.

\bibitem{Higham}
Nicholas~J Higham.
\newblock {\em Functions of matrices: {T}heory and computation}, volume 104.
\newblock SIAM, 2008.

\bibitem{hachimoji}
Shuichi Hoshika, Nicole~A Leal, Myong-Jung Kim, Myong-Sang Kim, Nilesh~B
  Karalkar, Hyo-Joong Kim, Alison~M Bates, Norman~E Watkins, Holly~A
  SantaLucia, Adam~J Meyer, et~al.
\newblock Hachimoji {DNA} and {RNA}: {A} genetic system with eight building
  blocks.
\newblock {\em Science}, 363(6429):884--887, 2019.

\bibitem{Mathematica}
Wolfram Research~{,} Inc.
\newblock Mathematica, {V}ersion 13.1.
\newblock 2022.

\bibitem{iosifescu2014finite}
Marius Iosifescu.
\newblock {\em Finite {M}arkov processes and their applications}.
\newblock Courier Corporation, 2014.

\bibitem{Cuthbert}
Cuthbert~R James.
\newblock The logarithm function for finite-state {M}arkov semi-groups.
\newblock {\em Journal of the London Mathematical Society}, 2(3):524--532,
  1973.

\bibitem{Jia}
Chen Jia.
\newblock A solution to the reversible embedding problem for finite markov
  chains.
\newblock {\em Statistics and Probability Letters}, 116:122–130, 2016.

\bibitem{Johansen}
S.~Johansen.
\newblock Some results on the imbedding problem for finite {M}arkov chains.
\newblock {\em Journal of the London Mathematical Society}, s2-8(2):345--351,
  1974.

\bibitem{kimura1957some}
Motoo Kimura.
\newblock Some problems of stochastic processes in genetics.
\newblock {\em The Annals of Mathematical Statistics}, pages 882--901, 1957.

\bibitem{kingman1962imbedding}
John Frank~Charles Kingman.
\newblock The imbedding problem for finite {M}arkov chains.
\newblock {\em Zeitschrift f{\"u}r Wahrscheinlichkeitstheorie und verwandte
  Gebiete}, 1(1):14--24, 1962.

\bibitem{leal2015transcription}
Nicole~A Leal, Hyo-Joong Kim, Shuichi Hoshika, Myong-Jung Kim, Matthew~A
  Carrigan, and Steven~A Benner.
\newblock Transcription, reverse transcription, and analysis of {RNA}
  containing artificial genetic components.
\newblock {\em ACS synthetic biology}, 4(4):407--413, 2015.

\bibitem{malyshev2012efficient}
Denis~A Malyshev, Kirandeep Dhami, Henry~T Quach, Thomas Lavergne, Phillip
  Ordoukhanian, Ali Torkamani, and Floyd~E Romesberg.
\newblock Efficient and sequence-independent replication of {DNA} containing a
  third base pair establishes a functional six-letter genetic alphabet.
\newblock {\em Proceedings of the National Academy of Sciences},
  109(30):12005--12010, 2012.

\bibitem{pachter2005algebraic}
Lior Pachter and Bernd Sturmfels.
\newblock {\em Algebraic statistics for computational biology}, volume~13.
\newblock Cambridge university press, 2005.

\bibitem{Roca}
Jordi Roca-Lacostena.
\newblock {\em The embedding problem for {M}arkov matrices}.
\newblock PhD thesis, Universitat Polit{\`e}cnica de Catalunya, May 2021.

\bibitem{RocaFernandez}
Jordi Roca-Lacostena and Jes{\'u}s Fern{\'a}ndez-S{\'a}nchez.
\newblock Embeddability of {K}imura {3ST} {M}arkov matrices.
\newblock {\em Journal of theoretical biology}, 445:128--135, 2018.

\bibitem{roca2018embeddability}
Jordi Roca-Lacostena and Jes{\'u}s Fern{\'a}ndez-S{\'a}nchez.
\newblock Embeddability of {K}imura 3st markov matrices.
\newblock {\em Journal of theoretical biology}, 445:128--135, 2018.

\bibitem{Runnenberg}
J.~Th. Runnenberg.
\newblock On {E}lfving's problem of imbedding a time-discrete markov chain in a
  time-continuous one for finitely many states.
\newblock {\em Proceedings of the KNAW - Series A, Mathematical Sciences},
  65:536–541, 1962.

\bibitem{schensted1958appendix}
Irene~V Schensted.
\newblock Appendix model of subnuclear segregation in the macronucleus of
  ciliates.
\newblock {\em The American Naturalist}, 92(864):161--170, 1958.

\bibitem{sismour2004pcr}
A~Michael Sismour, Stefan Lutz, Jeong-Ho Park, Michael~J Lutz, Paul~L Boyer,
  Stephen~H Hughes, and Steven~A Benner.
\newblock Pcr amplification of {DNA} containing non-standard base pairs by
  variants of reverse transcriptase from human immunodeficiency virus-1.
\newblock {\em Nucleic Acids Research}, 32(2):728--735, 2004.

\bibitem{YapPachter}
Yap V.~B. and Pachter L.
\newblock Identification of evolutionary hotspots in the rodent genomes.
\newblock {\em Genome Research}, 14(4):574--579, 2004.

\bibitem{weaver1985centrosymmetric}
James~R Weaver.
\newblock Centrosymmetric (cross-symmetric) matrices, their basic properties,
  eigenvalues, and eigenvectors.
\newblock {\em The American Mathematical Monthly}, 92(10):711--717, 1985.

\bibitem{Benner12nuclei}
Z~Yang, D~Hutter, P~Sheng, AM~Sismour, and SA~Benner.
\newblock Artificially expanded genetic information system: a new base pair
  with an alternative hydrogen bonding pattern.
\newblock {\em Nucleic Acids Research}, 34(21):6095–101, 2006.

\bibitem{yang2011amplification}
Zunyi Yang, Fei Chen, J~Brian Alvarado, and Steven~A Benner.
\newblock Amplification, mutation, and sequencing of a six-letter synthetic
  genetic system.
\newblock {\em Journal of the American Chemical Society}, 133(38):15105--15112,
  2011.

\bibitem{yang2007enzymatic}
Zunyi Yang, A~Michael Sismour, Pinpin Sheng, Nyssa~L Puskar, and Steven~A
  Benner.
\newblock Enzymatic incorporation of a third nucleobase pair.
\newblock {\em Nucleic acids research}, 35(13):4238--4249, 2007.

\end{thebibliography}
\smallskip

\printindex

    
 


\smallskip
\end{document}